\documentclass[12pt,reqno]{amsart}

\usepackage[latin1]{inputenc}
\usepackage{amsmath}
\usepackage{amsfonts}
\usepackage{amssymb}
\usepackage{graphics}
\usepackage{enumerate}
\usepackage{enumitem}
\usepackage{subfigure}
\usepackage{amssymb,amsmath,amsthm,amscd,epsf,latexsym,verbatim,graphicx,amsfonts}
\input epsf.tex

\usepackage{amscd}
\usepackage{amsmath}
\usepackage{amssymb}
\usepackage{amsthm}
\usepackage{epsf}
\usepackage{latexsym}
\usepackage{verbatim}
\input epsf.tex

\newtheorem*{PT}{Theorem A}

\theoremstyle{plain}
\newtheorem{theorem}{Theorem}[section]

\newtheorem{corollary}[theorem]{Corollary}
\newtheorem{lemma}[theorem]{Lemma}
\newtheorem{prop}[theorem]{Proposition}
\theoremstyle{definition}

\theoremstyle{remark}

\newcommand{\nri}{n\rightarrow\infty}

\newcommand{\bnri}{N\rightarrow\infty}
\newcommand{\bbZ}{\mathbb{Z}}
\newcommand{\bbR}{\mathbb{R}}
\newcommand{\bbC}{\mathbb{C}}

\newcommand{\bbN}{\mathbb{N}}

\newcommand{\mca}{\mathcal{A}}
\newcommand{\mcb}{\mathcal{B}}

\newcommand{\mce}{\mathcal{E}}

\newcommand{\mch}{\mathcal{H}}

\newcommand{\mcl}{\mathcal{L}}
\newcommand{\mcm}{\mathcal{M}}
\newcommand{\mcn}{\mathcal{N}}
\newcommand{\mco}{\mathcal{O}}

\newcommand{\mcv}{\mathcal{V}}
\newcommand{\mcw}{\mathcal{W}}

\newcommand{\arz}{a\rightarrow0^+}
\newcommand{\alrz}{\alpha\rightarrow0^+}

\newcommand{\rmp}{\rm{p}}
\newcommand{\cmp}{\rm{cp}}

\begin{document}

\title[]{Periodic Discrete Energy for Long-Range Potentials}

\date{\today}
\maketitle

\begin{center}
\textrm{D. P. Hardin \footnote{\label{note1}The research of the authors was supported, in part, by the National Science Foundation grant DMS-1109266}}, E. B. Saff \footnotemark[1], and B. Simanek \footnotemark[1]

\begin{small}Center for Constructive Approximation, Department of Mathematics, Vanderbilt University, Nashville, TN 37240
\vspace{.1cm}

e-mails: doug.hardin@vanderbilt.edu, edward.b.saff@vanderbilt.edu, brian.z.simanek@vanderbilt.edu \end{small}
\end{center}
\date{}

\begin{abstract}
We consider periodic energy problems in Euclidean space with a special emphasis on long-range potentials that cannot be defined through the usual infinite sum.  One of our main results builds on more recent developments of Ewald summation to define the periodic energy corresponding to a large class of long-range potentials.  Two particularly interesting examples are the logarithmic potential and the Riesz potential when the Riesz parameter is smaller than the dimension of the space.  For these examples, we use analytic continuation methods to provide concise formulas for the periodic kernel in terms of the Epstein Hurwitz Zeta function.  We apply our energy definition to deduce several properties of the minimal energy including the asymptotic order of growth and the distribution of points in energy minimizing configurations as the number of points becomes large.  We conclude with some detailed calculations in the case of one dimension, which shows the utility of this approach.
\end{abstract}

\vspace{4mm}


\vspace{2mm}


\vspace{2mm}

\normalsize

\section{Introduction}\label{intro}

For an $N$-tuple $\omega_N=(x_j)_{j=1}^N$ of points confined to a compact subset $\Omega_0\subseteq\bbR^d$, we define its $f$-\textit{energy} as
\begin{align}\label{endef}
E_f(\omega_N):=\sum_{i\neq j}f(x_i-x_j),
\end{align}
where $f$ is a lower-semicontinuous function from $\bbR^d$ to $\bbR\cup\{\infty\}$.  The study of minimal energy investigates configurations that minimize this energy among all such $N$-tuples.  Therefore, we define
\begin{align}\label{minendef}
\mce_f(\Omega_0,N):=\inf_{\omega_N\in\Omega^N}E_f(\omega_N).
\end{align}
The lower semi-continuity of the function $f$ implies that minimizers exist and so the infimum in (\ref{minendef}) is in fact a minimum.  In recent years, there has been much interest in studying the asymptotics of $\mce_f(\Omega_0,N)$ as $N$ becomes large and deducing properties of the energy minimizing configurations (see \cite{LowComp,BHS,BHS2,HSNot,HSAdv,KS}).  Of particular interest is the case of a Riesz potential, where $f(w)=|w|^{-s}$ and $|\cdot|$ denotes the Euclidean norm.

We will consider the energy problem in a related setting, which includes additional symmetry that will simplify many of our computations.  Let $\{v_1,\ldots,v_d\}$ be a collection of $d$ linearly independent vectors in $\bbR^d$ and let $V$ be the $d\times d$ matrix whose $j^{th}$ column is equal to $v_j$.  We set
\[
\Omega:=\left\{w:w=\sum_{j=1}^d\alpha_jv_j,\, \alpha_j\in[0,1),\, j=1,2,\ldots,d\right\},
\]
and we will denote its closure in $\bbR^d$ by $\overline{\Omega}$.
Let $\mcv$ be the lattice determined by the matrix $V$; that is, $\mcv:=\{Vk:k\in\bbZ^d\}$ and let $\mcv^*$ be the lattice dual to $\mcv$; that is, $\mcv^*=\left\{w\in\bbR^d:w\cdot v\in\bbZ\mbox{ for all }v\in\mcv\right\}$.  We can think of $\Omega$ as a fundamental cell of the quotient space $\bbR^d/\mcv$, and we highlight the fact that in the quotient topology, $\Omega$ is compact.

If $f$ is a lower-semicontinuous function from $\bbR^d$ to $\bbR\cup\{\infty\}$ that decays sufficiently quickly at infinity, then we define the \textit{classical periodic $f$-energy} of a configuration $\omega_N=(x_j)_{j=1}^N\in(\bbR^d)^N$ by
\begin{align}\label{periodf}
E_f^{\cmp}(\omega_N):=\sum_{i\neq j}\left(\sum_{v\in\mcv}f(x_i-x_j+v)\right).
\end{align}
In this context, the function $f$ is referred to as the \textit{potential function}.  If $\mca\subseteq\Omega$ is compact (in the quotient topology on $\bbR^d/\mcv$) and infinite, then we define
\begin{align}\label{pminendef}
\mce_f^{\cmp}(\mca,N):=\inf_{\omega_N\in\mca^N}E_f^{\cmp}(\omega_N).
\end{align}
The physical interpretation of the energy (\ref{periodf}) is easy to describe.  Consider a crystal that consists of a particular configuration of particles that is confined to a compact set and this configuration is repeated in a periodic fashion throughout a very large region of space.  If the particles exhibit a repelling force on one another, they will arrange themselves in a manner that minimizes the energy of the entire crystal.  To approximate the energy of this crystal, it suffices to approximate the energy of one cell of the crystal lattice and then multiply by the number of cells.  When calculating the energy of a single cell, we make the further approximation that the lattice is infinite, so we must sum up the contribution to the energy of the interaction between every particle in the cell and every other particle in the entire crystal.  When the interaction between the particles $x$ and $y$ is given by $f(x-y)$, the resulting sum is of the form (\ref{periodf}).

We should point out that some authors define the periodic energy using different notation (see \cite{CKJAMS}) by choosing $x_j\in\bbR^d$ for $j\in\{1,\ldots,N\}$, defining the set $\Lambda$ by
\[
\Lambda:=\bigcup_{j=1}^N\left\{x_j+v:v\in\mcv\right\},
\]
and then defining the periodic energy by
\[
\sum_{k=1}^N\left(\sum_{{q\in\Lambda}\atop{q\neq x_k}}f(x_k-q)\right).
\]
It is easy to see that this sum and (\ref{periodf}) differ by the so-called \textit{self-energy} term, which takes the form
\[
N\cdot\sum_{v\in\mcv\setminus\{0\}}f(v).
\]
In the specific case of a Coulomb potential, an alternating sum similar to this self-energy term is related to the Madelung constant, which is of significant interest in its own right (see \cite{BBP,BBT}).  Since the self-energy term is independent of the points in the configuration, its presence does not meaningfully effect the asymptotics of $\mce_{f}^{\cmp}(\mca,N)$ for large $N$, so its inclusion or omission is not relevant for our investigation.

Of course, the sum (\ref{periodf}) will not converge without the decay assumption on the function $f$, so we will introduce a {\em renormalized} energy given by (\ref{e1def}) and (\ref{1def}) below to compute the energy of a configuration for a broader class of potentials.  We provide a derivation of the kernel formula (\ref{1def}) in Section \ref{cf2}, and describe its relation to formulas that have previously appeared in the physics literature (see for example \cite{MP,PPdL}).

The problem of summing divergent or conditionally convergent series related to physical phenomena has a long history.  One of the most widely used methods is known as Ewald summation (see \cite{Ewald}), which is a method for defining Coulomb (that is, electrostatic) energies.
Various improvements of the Ewald summation method have arisen since the original paper \cite{Ewald}.  Indeed, the recent advances in computational mathematics have inspired many faster and more stable algorithms related to lattice summation (see, for example, \cite{BG,FSW,FreSmit,GR,NaCe,PPdL}).  There have also been improvements to the scope of the Ewald method.  In \cite{Heyes}, Heyes studied the effect of utilizing different charge distribution functions in the Ewald method and in \cite{PPdL} the Ewald method is applied to a large collection of potentials that includes the Coulomb interaction.  We also note that the recent methods of \cite{SanSer}  can be utilized to define such `renormalized' energies for infinite point configurations (e.g., the periodic case) interacting through the Coulomb potential in two-dimensions.

 Compared with the physics literature, extensive results in the mathematics literature on periodic discrete energy are more difficult to find.  Some analytic methods for evaluating conditionally convergent sums can be found in \cite{BBS} and rigorous results concerning the Madelung constant appear in \cite{BBP,BBT,Terras}.  One of our goals is to define an energy functional that admits a mathematically rigorous derivation, has certain desirable properties (see Theorems \ref{main} and \ref{intasymp}), and generalizes the ideas presented in many of the aforementioned papers.

Before we state the definition of our energy functional, we need to specify the potential functions that we will consider.  If $\nu$ is a signed measure, we will denote by $\nu^+$ and $\nu^-$ its positive and negative parts, respectively.

\vspace{2mm}

\noindent\textbf{Definition 1.}  We will say that a lower-semicontinuous function $f:\bbR^d\rightarrow\bbR\cup\{\infty\}$ is a \textit{G-type potential} if it satisfies the following property:
\begin{enumerate}
\item[(G)]  for every $q\in\bbR^d\setminus\{0\}$, $f(q)$ is finite and can be expressed as
\[
f(q)=\int_0^{\infty}e^{-|q|^2t}\textrm{d}\mu_f(t),
\]
for some signed measure $\mu_f$ on $(0,\infty)$ having finite negative part.  We also define $f(0):=\mu_f(\bbR^d)$, which exists as an element of $\bbR\cup\{\infty\}$.
\end{enumerate}
We will say that a lower-semicontinuous function $f:\bbR^d\rightarrow\bbR\cup\{\infty\}$ is a \textit{weak G-type potential} if there is a function $f^*:(0,1)\rightarrow\bbR$ and a signed measure $\mu_f$ on $(0,\infty)$ with finite negative part so that the following conditions are satisfied:
\begin{enumerate}
\item[(W1)]
for every $q\in\bbR^d\setminus\{0\}$, $f(q)$ is finite and can be expressed as
\[
f(q)=\lim_{\alrz}\left(\int_{\alpha}^{\infty}e^{-|q|^2t}\textrm{d}\mu_f(t)+f^*(\alpha)\right),
\]
where $\int_{\alpha}^{\infty}e^{-|q|^2t}\textrm{d}\mu_f(t)<\infty$ for all $\alpha>0$, and
\vspace{2mm}

\item[(W2)]
If $w_0$ is an element of $\mcv^*\setminus\{0\}$ of minimal length, then
\[
\int_0^1\frac{e^{-\pi^2|w_0|^2/t}}{t^{d/2}}\textrm{d}\mu_f(t)<\infty.
\]

\end{enumerate}

\vspace{4mm}

The terminology ``G-type potential" is short for Gaussian-type potential in that we are expressing the potentials $f$ in the form $f(q)=F(|q|^2)$, where $F$ is the Laplace transform of a signed measure on $(0,\infty)$.  If $\mu_f$ is positive, then its Laplace transform is a completely monotone function from $(0,\infty)$ to itself\footnote{A function $F$ is said to be \textit{completely monotone} on $(0,\infty)$ if $(-1)^kF^{(k)}(x)\geq0$ holds on $(0,\infty)$ for every $k\in\{0,1,2,\ldots\}$.}.  Therefore, G-type potentials are defined via the difference of two completely monotone functions on $(0,\infty)$ and weak G-type potentials are renormalized limits of G-type potentials.

Whenever we refer to a weak G-type potential $f$, we will associate to it the measure $\mu_f$ that appears in the definition.  It is clear that every G-type potential is a weak G-type potential, but the converse is false.  An example of a weak G-type potential that is not a G-type potential is the logarithm $f(x)=-\log(|x|^2)$.  Indeed, the logarithm motivates our definition of weak G-type potentials and it is true that
\begin{align}\label{logiden}
-\log(r^2)=\lim_{\alrz}\left(\int_{\alpha}^{\infty}\frac{e^{-r^2t}}{t}\textrm{d}t+\gamma+\log\alpha\right),
\end{align}
where $\gamma$ is the Euler-Mascheroni constant (see \cite[Equation 3.77]{Gamma}).  The corresponding function $f^*$ is given by $f^*(\alpha)=\gamma+\log\alpha$ and the corresponding measure $\mu_f$ is $t^{-1}\textrm{d}t$.  We will see some examples of G-type potentials in Section \ref{exam}, such as the Riesz potential $f(x)=|x|^{-s}$ where $s>0$.

With these preliminaries, we now present a definition that will be of fundamental importance to the remainder of the paper.

\vspace{2mm}

\noindent\textbf{Definition 2.}  Let $f$ be a weak G-type potential with corresponding measure $\mu_f$.  Assume that the matrix $V$ that determines $\mcv$ satisfies $\det(V)=1$ and consider $\omega_N=(x_j)_{j=1}^N\in(\bbR^d)^N$.  We define the \textit{periodic $f$-energy of $\omega_N$ associated with the lattice $\mcv$} by
\begin{align}\label{e1def}
E_{f}^{\rmp}(\omega_N):=\sum_{{1\leq k,j\leq N}\atop{k\neq j}}K_f^{\rmp}(x_j,x_k),
\end{align}
where
\begin{align}\label{1def}
K_f^{\rmp}(x,y):=\sum_{v\in\mcv}\int_1^{\infty}e^{-|x-y+v|^2t}\textrm{d}\mu_f(t)+\sum_{w\in\mcv^*\setminus\{0\}}e^{2\pi iw\cdot(x-y)}\int_0^1\frac{\pi^{d/2}}{t^{d/2}}e^{-\pi^2|w|^2/t}\textrm{d}\mu_f(t).
\end{align}
We also define
\begin{align}\label{minendef2}
\mce_f^{\rmp}(\mca,N):=\inf_{\omega_N\in\mca^N}E_f^{\rmp}(\omega_N),
\end{align}
where $\mca\subseteq\overline{\Omega}$ is infinite.

\vspace{2mm}

\noindent\textit{Remark.}  When we write $\int_a^bh(t)\textrm{d}\mu(t)$ for any measure $\mu$, we mean the integral over the half-open interval $[a,b)$.

\vspace{2mm}

\noindent\textit{Remark.}  We allow for the possibility of a configuration having infinite energy, but this can only happen if $x_i-x_j\in\mcv$ for some $i\neq j$.

\vspace{2mm}

The formula (\ref{1def}) arises from a renormalization process involving limits of classical periodic energy functionals.  Namely, we derive (\ref{1def}) by first modifying the potential so that the sum (\ref{periodf}) for the modified potential converges, and then continuously remove this added decay by pushing it out to infinity and renormalizing the sum in a way that is independent of the configuration.  Further details are provided in Section \ref{cf2}.

The focus of this paper will be on applications of (\ref{e1def}) and (\ref{1def}) to minimal energy problems.  Before we apply Definition 2, we list some of its properties.  The following theorem shows that (\ref{e1def}-\ref{1def}) has several properties one would expect from a periodic energy definition.

\begin{theorem}\label{main}
If $f$ is a weak G-type potential, then its kernel has the  following properties:
\begin{enumerate}[before=\itshape,font=\normalfont]
\item\label{welldef}   $K_f^{\rmp}$ is  well defined and continuous as a function from $\bbR^d\times \bbR^d$ to $\bbR\cup \{\infty\}$.  Furthermore, $K_f^{\rmp}(x,y)$ is finite for any   $x,y\in \bbR^d$ such that  $x -y\not\in\mcv$.

  \item\label{periodice}    $K_f^{\rmp}(x,y)$  is symmetric, periodic in each coordinate with respect to the lattice $\mcv$ and depends only on $x-y$.
\item\label{samesum} If $f$ is a G-type potential and the sum (\ref{periodf}) converges absolutely, then $E_f^{\cmp}$ and $E_f^{\rmp}$ differ by a constant multiple of $N(N-1)$, where the constant does not depend on the configuration.
\end{enumerate}
\end{theorem}

\noindent\textit{Remark.} As shown in the proof, if $ \mu_f^+([1,\infty))<\infty$, then  $K_f^{\rmp}(x,y)$ is also finite for $x-y\in \mcv$,  otherwise $K_f^{\rmp}(x,y)=\infty$   for $x-y\in \mcv$.
 A configuration $(x_j)_{j=1}^N$  will be called \textit{non-degenerate} if $x_j-x_k\not\in \mcv$ for any $j\neq k$ and so the energy in \eqref{e1def} of such a configuration must be finite.

\vspace{2mm}

\begin{proof}[Proof of Theorem \ref{main}(a)]
By assumption on $f$, we have $\mu_f=\mu_f^+ -\mu_f^-$ for some positive $\mu_f^{+}$    and some finite positive measure $\mu_f^-$.     We shall begin by establishing that   the second sum in (\ref{1def}) converges uniformly on $\bbR^d\times \bbR^d$ by verifying that   the sum of the integrals converges absolutely.    From the representation
\begin{align*}
\sum_{w\in\mcv^*\setminus\{0\}}\frac{e^{-\pi^2|w|^2/t}}{t^{d/2}}&=\frac{e^{-\pi^2|w_0|^2/t}}{t^{d/2}}\sum_{w\in\mcv^*\setminus\{0\}}e^{\pi^2(|w_0|^2-|w|^2)/t}
\end{align*}
and the fact that the last sum is increasing in $t$ and  converges when $t=1$, it follows that   the sum of the integrands converges  and is bounded on $[0,1]$ by a constant multiple of $e^{-\pi^2|w_0|^2/t}t^{-d/2}$.
Consequently, applying condition (W2), we obtain that the second sum converges to a finite continuous function on
$\bbR^d\times \bbR^d$.

Next we consider the first sum in  (\ref{1def}).
Since
$$\sum_{v\in\mcv} e^{-|x-y+v|^2},$$
converges uniformly for $(x,y)\in\bbR^d\times \bbR^d$ it follows that
\[
\sum_{v\in\mcv}\int_1^{\infty}e^{-|x-y+v|^2t}\textrm{d}\nu(t),
\]
also converges uniformly on $\bbR^d\times \bbR^d$ if $\nu$ is a finite measure. Thus, if $\mu_f^+([1,\infty))$ is finite, we have that $K_f^p$ is continuous and finite on  $\bbR^d\times \bbR^d$.

Finally, we consider the case that $\mu_f^+([1,\infty))=\infty$.  Let $x,y\in \bbR^d$ be such that $x-y\not\in \mcv$ and
choose $\delta$ so that $0<\delta<|x -y+v|^2$ for all $v\in\mcv$.  Define
\[
\Theta_\mcv(t):=\sum_{v\in\mcv}e^{(\delta-|x-y+v|^2)t},
\]
and observe from the finiteness property of (W1) that
\begin{align*}
\sum_{v\in\mcv}\int_1^{\infty}e^{-|x-y+v|^2t}\textrm{d}\mu_f^+(t)&=\int_1^{\infty}\sum_{v\in\mcv}e^{(\delta-|x-y+v|^2)t}e^{-\delta t}\textrm{d}\mu_f^+(t)\\
&=\int_1^{\infty}\Theta_{\mcv}(t)e^{-\delta t}\textrm{d}\mu_f^+(t)<\infty,
\end{align*}
since $\Theta_{\mcv}(t)$ is bounded and decreasing on $[1,\infty)$.  This establishes convergence of the first sum in (\ref{1def}).  It is not difficult to show that the convergence is uniform on any closed subset of $\mathcal{D}'=\{(x,y)\in \bbR^d\times \bbR^d: x-y\not\in \mcv\}$ and so the first sum is continuous on $\mathcal{D}'$.

Fix $v\in \mcv$.  For  $x-y$ in a sufficiently small neighborhood of $-v$, the dominant term in the first sum in  (\ref{1def}) is $h(x-y):=\int_1^{\infty}e^{-|x-y+v|^2t}\textrm{d}\mu_f^+(t)$ while the remainder is continuous and finite for $x-y$ in this neighborhood.  Since $\mu_f^+([1,\infty))=\infty$, $h(x-y)\to \infty$ as $x-y\to -v$.     Consequently, $K_f^p$ is continuous as a function from $\bbR^d\times \bbR^d$ to $\bbR\cup\{\infty\}$.
\end{proof}

\begin{proof}[Proof of Theorem \ref{main}(b)]
The symmetry and periodicity of the kernel is clear from the form of the kernel and the definition of the dual lattice.
\end{proof}

\noindent\textit{Remark.}  We postpone the proof of Theorem \ref{main}(c) until Section \ref{cf2}.

\vspace{2mm}

One of our goals is  to investigate the asymptotics of the minimal energy (as defined in (\ref{minendef2})) as $N$ becomes large.  One of our results (see Theorem \ref{intasymp} below) states that if $\mu_f$ is positive (more generally, if the kernel $K_f^{\rmp}$ is integrable), then the limit:
\[
\lim_{\bnri}\frac{\mce_f^{\rmp}(\mca,N)}{N^2},
\]
 exists, is finite, and can be expressed as an explicit integral provided $\mca$ satisfies some additional hypotheses.  We will apply this result to determine the leading order of growth of the minimal periodic energy corresponding to the potential function $f_s(x):=|x|^{-s}$ for all values of $s\in(0,d)$ when $\mca=\Omega$ (see Corollary \ref{slessdexam}).  When $s\geq d$, we will show that the leading order of growth is the same as in the non-periodic setting, even if $\mca\neq\Omega$ (see Theorem \ref{equalproblem}).  This is not surprising because for large values of $s$, it is the nearest neighbor interactions that dominate the asymptotics, so the periodization of the problem should only have a slight effect.

In the next section, we will investigate minimal energy asymptotics for positive integrable kernels.  In Section \ref{exam}, will study the resulting kernels and the minimal energy asymptotics for Riesz and log-Riesz potentials and also introduce a convenient formula for the periodic logarithmic kernel.  In Section \ref{cf2} we will provide the details of our derivation of the formula (\ref{1def}) and show that it arises naturally from a certain renormalization process.  We will place a particular emphasis on the robust nature of our derivation and show that many different approaches to defining a periodic energy yield the same result.  Section \ref{d1} contains some detailed minimal energy calculations for several potentials - including the Riesz potential - in the one-dimensional setting.  These results are extremely precise and highlight the possible advantages of considering the periodic problem when studying minimal energy configurations.

\subsection{Notation.}
Throughout the remainder of the paper, we will use $x_{jk}$ to mean $x_j-x_k$.  We will always assume that our lattice $\mcv$ is determined by a matrix $V$ satisfying $\det(V)=1$; i.e., the {\em co-volume of $\mcv$} is 1.  This can always be achieved by an appropriate rescaling of the lattice and will simplify some of our formulas.  For any integrable function $h$, we denote its Fourier transform by $\hat{h}$, that is,
\[
\hat{h}(y):=\int_{\bbR^d}e^{-2\pi iy\cdot t}h(t)dt.
\]
If $\nu$ is a signed measure, we will write $\hat{\nu}$ to denote its Fourier transform $\displaystyle{\hat{\nu}(y):=\int_{\bbR^d}e^{-2\pi iy\cdot t}d\nu(t)}$.  We will use $\mch_q(X)$ to denote the $q$-dimensional Hausdorff measure of a set $X$.

\section{Integrable Kernels}\label{integrable}

In this section, we will fix $\mca\subseteq\Omega$ to be an infinite set that is compact in the quotient topology on $\bbR^d/\mcv$.  Let $\mcm_{+,1}(\mca)$ be the collection of all positive probability measures with support in $\mca$, where we define the support of the measure in the topology of $\bbR^d/\mcv$.
Our goal in this section is to prove the following pair of theorems:

\begin{theorem}\label{intthrm}
Suppose $f$ is a weak G-type potential and there exists some $\lambda\in\mcm_{+,1}(\mca)$ satisfying
\begin{align}\label{elldef}
\normalfont\mcl_f(\lambda):=\int\int K^{\rmp}_f(x,y)\textrm{d}\lambda(x)\textrm{d}\lambda(y)<\infty.
\end{align}
If the signed measure $\mu_f$ associated with $f$ satisfies
\begin{align}\label{gfin}
\normalfont\int_0^{\infty}\frac{\pi^{d/2}}{t^{d/2}}e^{-\pi^2|w|^2/t}\textrm{d}\mu_f(t)>0,\qquad \mbox{ for all }w\in\mcv^*\setminus\{0\},
\end{align}
then   the set
\begin{align}\label{minset}
\left\{\lambda\in\mcm_{+,1}(\mca):\mcl_f(\lambda)=\inf_{\nu\in\mcm_{+,1}(\mca)}\mcl_f(\nu)\right\},
\end{align}
consists of a single element  (denoted by $\nu_f$).

In particular, if $\mathcal{A}=\Omega$, then $\nu_f$ is $d$-dimensional Lebesgue measure restricted to $\Omega$
and
\begin{align}\label{ellA}
\normalfont\mcl_f(\nu_f):=\int\int K^{\rmp}_f(x,y)\textrm{d}x\textrm{d}y=\int_1^\infty \frac{\pi^{d/2}}{t^{d/2}}\textrm{d}\mu_f(t).
\end{align}
\end{theorem}



\noindent\textit{Remark.}  It is clear that the condition (\ref{gfin}) is satisfied if $\mu_f^-=0$.

\vspace{2mm}

\begin{theorem}\label{intasymp}
Suppose $f$ is a weak G-type potential.
\begin{itemize}[before=\itshape,font=\normalfont]
\item[I)] If (\ref{elldef}) holds for some $\lambda\in\mcm_{+,1}(\mca)$ and $\mu_f$ satisfies (\ref{gfin}), then
\begin{align}\label{n2lim}
\normalfont\lim_{\bnri}\frac{\mce_f^{\rmp}(\mca,N)}{N^2}=\int_{\mca}\int_{\mca}K_f^{\rmp}(x,y)\textrm{d}\nu_f(x)\textrm{d}\nu_f(y),
\end{align}
where $\nu_f$ is the unique element of the set (\ref{minset}).
\item[II)]  If $\mu_f$ satisfies (\ref{gfin}) but $\mcl_f(\lambda)=\infty$ for all $\lambda\in\mcm_{+,1}(\mca)$, then the limit in (\ref{n2lim}) is positive infinity.
\end{itemize}
\end{theorem}

Theorem \ref{intasymp}(I) tells us that if the kernel $K_f^{\rmp}$ does not blow up too quickly along the diagonal of $\mca\times\mca$, then we can write down the leading term in the asymptotic expansion of $\mce_f^{\rmp}(\mca,N)$.  This conclusion will also have implications for the macroscopic distribution as $\bnri$ of minimal energy configurations (which exist because $\mca$ is compact in $\bbR^d/\mcv$; see Corollary \ref{zeros} below).

It will be no trouble to prove Theorem \ref{intthrm} (using standard machinery) once we have established the following result:

\begin{theorem}\label{posdefker}
Suppose $f$ is a weak G-type potential and $\mu_f$ satisfies (\ref{gfin}).  If $\lambda=\lambda_1-\lambda_2$, where each $\lambda_i\in\mcm_{+,1}(\mca)$ and
\begin{align}\label{finite energy}
\normalfont\int_{\mca}\int_{\mca}K_f^{\rmp}(x,y)\,\textrm{d}\lambda_i(x)\textrm{d}\lambda_i(y)<\infty,\qquad\qquad i=1,2
\end{align}
then
\begin{align}\label{posconc}
\normalfont\int_{\mca}\int_{\mca}K_f^{\rmp}(x,y)\,\textrm{d}\lambda(x)\textrm{d}\lambda(y)\geq0,
\end{align}
with equality if and only if $\lambda$ is the zero measure.
\end{theorem}

The proof will rely on our next lemma involving the Fourier transform.
\begin{lemma}\label{noz}
Let $\gamma$ be a signed measure on $\mca$ that can be written as the difference of two members of $\mcm_{+,1}(\mca)$.  If the Fourier transform $\hat{\gamma}(w)=0$ for all $w\in\mcv^*$, then $\gamma$ is the zero measure.
\end{lemma}

\begin{proof}
Throughout this proof, we keep in mind that $\mca$ is a (compact) subset of $\bbR^d/\mcv$.
Define
\[
\mcw:=\left\{k(x)=\sum_{j=0}^ma_je^{2\pi iw_j\cdot x}:a_j\in\bbR,\, w_j\in\mcv^*,\, m\in\bbN_0\right\}.
\]
Let $\overline{\mcw}$ be the closure of $\mcw$ in the uniform norm on $\mca$.  It is easily seen that $\overline{\mcw}$ is a closed algebra of continuous functions on $\mca$ that includes the constant functions and separates points.  The Stone-Weierstrass Theorem tells us that $\overline{\mcw}$ is all continuous functions on $\mca$.  Our hypotheses imply that $\gamma(k)=0$ for all $k\in\overline{\mcw}$, and hence $\gamma$ is the zero measure.
\end{proof}

\begin{proof}[Proof of Theorem \ref{posdefker}]
For the purpose of future reference, we initially only assume that $\lambda_1$ and $\lambda_2$  are positive finite measures on $\mca$ satisfying \eqref{finite energy}; i.e, we postpone the assumption that $\lambda(\mca)=\lambda_1(\mca)-\lambda_2(\mca)=0$ until later in the proof.
First notice that the proof of Theorem \ref{main}(a) shows that the sum over $\mcv^*$ in (\ref{1def}) is uniformly bounded in $x$ and $y$.  Therefore, we may apply the Fubini Theorem and switch the infinite sum with the integral to obtain
\begin{align}
\nonumber&\int_{\mca}\int_{\mca}\left(\sum_{w\in\mcv^*\setminus\{0\}}e^{2\pi iw\cdot(x-y)}\int_0^1\frac{\pi^{d/2}}{t^{d/2}}e^{-\pi^2|w|^2/t}\textrm{d}\mu_f(t)\right)\textrm{d}\lambda(x)\textrm{d}\lambda(y)\\
\label{wpart}&\qquad\qquad=\sum_{w\in\mcv^*\setminus\{0\}}|\hat{\lambda}(w)|^2\int_0^1\frac{\pi^{d/2}}{t^{d/2}}e^{-\pi^2|w|^2/t}\textrm{d}\mu_f(t).
\end{align}

Consider now the sum over $\mcv$ in (\ref{1def}).  Given a measure $\lambda$ as in the statement of the theorem, define $G_t(x):=e^{-t|x|^2}$ and $h_t(x):=\widehat{G}_t(x)|\hat{\lambda}(x)|^2=\widehat{G}_t(x)\hat{\lambda}(x)\hat{\lambda}(-x)$.  Since $\lambda$ has bounded support, it is easily seen that $\hat{\lambda}$ is infinitely differentiable and every derivative of $\hat{\lambda}$ is a bounded function on $\bbR^d$.  Recall that $\widehat{G}_t$ is  a Gaussian; i.e.,
\begin{equation}\label{GaussFT}
\widehat{G}_t(y)=\frac{\pi^{d/2}}{t^{d/2}}e^{-\pi^2|y|^2/t}, \qquad (y\in \bbR^d),
\end{equation} and so  $h_t(x)$ is a Schwartz function for every $t>0$.

Notice that $G_t*\lambda$ and its Fourier transform are in $L^1(\bbR^d)$ (see \cite[Proposition 8.49]{Folland}), so we may use the Fourier inversion formula and the Fubini Theorem to see that for any fixed $v\in\mcv$ and any fixed $t>0$ we have
\begin{align}
\nonumber\int_{\mca}\int_{\mca} e^{-|x-y+v|^2t}\textrm{d}\lambda(x)\textrm{d}\lambda(y)&=\int_{\mca}(G_t*\lambda)(x+v)\,\textrm{d}\lambda(x)\\
\nonumber&=\int_{\mca}\int_{\bbR^d} e^{2\pi i(x+v)\cdot q}(\widehat{G_t*\lambda})(q)\,\textrm{d}q\,\textrm{d}\lambda(x)\\
\nonumber&=\int_{\bbR^d}\int_{\mca} e^{2\pi i(x+v)\cdot q}(\widehat{G_t*\lambda})(q)\,\textrm{d}\lambda(x)\,\textrm{d}q\\
\nonumber&=\int_{\bbR^d} e^{2\pi iv\cdot q}h_t(q)\,\textrm{d}q\\
\label{hatv}&=\hat{h}_t(-v).
\end{align}

Let us split $\mcv$ into two subsets.  We define $\mcv_1$ to be all $v\in\mcv$ such that there exist two points $a,b\in\overline{\mca}$ with $a-b=v$ and we set $\mcv_2:=\mcv\setminus\mcv_1$ (here $\overline{\mca}$ means the closure of $\mca$ in $\bbR^d$).  As in the proof of Theorem \ref{main}(a), it is straightforward to show that
\[
\sum_{v\in\mcv_2}\int_1^{\infty}e^{-|x-y+v|^2t}\textrm{d}\mu_f(t)
\]
is uniformly bounded in $x,y\in\mca$.  Therefore, we may change the order of integration and summation and write
\begin{align}\label{v2part}
\int_{\mca}\int_{\mca}\left(\sum_{v\in\mcv_2}\int_1^{\infty}e^{-|x-y+v|^2t}\textrm{d}\mu_f(t)\right)\textrm{d}\lambda(x)\textrm{d}\lambda(y)=\int_1^{\infty}\sum_{v\in\mcv_2}\hat{h}_t(-v)\textrm{d}\mu_f(t),
\end{align}
which we know is finite.

It remains to deal with the finite collection $\mcv_1$.  For any $v\in\mcv_1$, the finiteness of $\mu_f^-$ implies that there is a constant $\sigma>0$ so that
\[
\mcl_f(\lambda_i)\geq\int_1^{\infty}\int_{\mca}\int_{\mca} e^{-|x-y+v|^2t}\textrm{d}\lambda_i(x)\textrm{d}\lambda_i(y)\textrm{d}\mu_f^+(t)-\sigma,\qquad i=1,2,
\]
where $\mcl_f$ is defined as in (\ref{elldef}).  Therefore, when we write
\begin{align*}
&\int_{\mca}\int_{\mca}\int_1^{\infty}e^{-|x-y+v|^2t}\textrm{d}\mu_f^+(t)\textrm{d}\lambda(x)\textrm{d}\lambda(y)=\int_{\mca}\int_{\mca}\int_1^{\infty}e^{-|x-y+v|^2t}\textrm{d}\mu_f^+(t)\textrm{d}\lambda_1(x)\textrm{d}\lambda_1(y)\\
&\qquad+\int_{\mca}\int_{\mca}\int_1^{\infty}e^{-|x-y+v|^2t}\textrm{d}\mu_f^+(t)\textrm{d}\lambda_2(x)\textrm{d}\lambda_2(y)-2\int_{\mca}\int_{\mca}\int_1^{\infty}e^{-|x-y+v|^2t}\textrm{d}\mu_f^+(t)\textrm{d}\lambda_1(x)\textrm{d}\lambda_2(y)\\
&\quad=\int_1^{\infty}\int_{\mca}\int_{\mca}e^{-|x-y+v|^2t}\textrm{d}\lambda_1(x)\textrm{d}\lambda_1(y)\textrm{d}\mu_f^+(t)+\int_1^{\infty}\int_{\mca}\int_{\mca}e^{-|x-y+v|^2t}\textrm{d}\lambda_2(x)\textrm{d}\lambda_2(y)\textrm{d}\mu_f^+(t)\\
&\qquad\qquad\qquad\qquad-2\int_1^{\infty}\int_{\mca}\int_{\mca}e^{-|x-y+v|^2t}\textrm{d}\lambda_1(x)\textrm{d}\lambda_2(y)\textrm{d}\mu_f^+(t),
\end{align*}
we know that the first two terms in this last expression are finite, while the third is in $[-\infty,0]$.  Therefore, for each $v\in\mcv_1$, we have
\begin{align*}
\int_1^{\infty}\hat{h}_t(-v)\textrm{d}\mu_f^+(t)\in[-\infty,\infty).
\end{align*}
Furthermore, from (\ref{hatv}), we obtain
\begin{align*}
\int_1^{\infty}\hat{h}_t(-v)\textrm{d}\mu_f^-(t)=\int_1^{\infty}\int_{\mca}\int_{\mca}e^{-|x-y+v|^2t}\textrm{d}\lambda(x)\textrm{d}\lambda(y)\textrm{d}\mu_f^-(t)\in(-\infty,\infty),
\end{align*}
because $\mu_f^-$ is finite.  Consequently,
\begin{align*}
\int_1^{\infty}\sum_{v\in\mcv_1}\hat{h}_t(-v)\textrm{d}\mu_f(t)\in[-\infty,\infty),
\end{align*}
which - together with (\ref{v2part}) - implies
\[
\int_1^{\infty}\sum_{v\in\mcv}\hat{h}_t(v)\textrm{d}\mu_f(t)
\]
exists as an extended real number in $\bbR\cup\{-\infty\}$.

Since $h_t$ is a Schwartz function, we may apply the Poisson summation formula (see Appendix A) to conclude that
\begin{align}\label{vsumpart}
\int_1^{\infty}\sum_{v\in\mcv}\hat{h}_t(v)\textrm{d}\mu_f(t)=\int_1^{\infty}\sum_{w\in\mcv^*}h_t(w)\textrm{d}\mu_f(t)=\int_1^{\infty}\sum_{w\in\mcv^*}|\hat{\lambda}(w)|^2\widehat{G}_t(w)\textrm{d}\mu_f(t).
\end{align}
In order to show that this quantity is finite, it suffices to show that it is not negative infinity, and for this it is enough to consider the integral with respect to $\mu_f^-$.  Indeed, another application of Poisson summation shows
\begin{align*}
\int_1^{\infty}\sum_{w\in\mcv^*}|\hat{\lambda}(w)|^2\widehat{G}_t(w)\textrm{d}\mu_f^-(t)&=\int_1^{\infty}\sum_{v\in\mcv}\int_{\mca}\int_{\mca}e^{-|x-y+v|^2t}\textrm{d}\lambda(x)\textrm{d}\lambda(y)\textrm{d}\mu_f^-(t),
\end{align*}
which is clearly finite because $\mu_f^-$ is finite.  Therefore, the integrand on the far right-hand side of (\ref{vsumpart}) is integrable with respect to $|\mu_f|$, which means we may bring the sum outside of the integral.  Combining (\ref{wpart}) and (\ref{vsumpart}), we obtain
\begin{align}\label{lastsum}
\iint K_f^{\rmp}(x,y)\textrm{d}\lambda(x)\textrm{d}\lambda(y)=\sum_{w\in\mcv^*\setminus\{0\} }\int_0^{1}|\hat{\lambda}(w)|^2\widehat{G}_t(w)\textrm{d}\mu_f(t)+\sum_{w\in\mcv^*}\int_1^{\infty}|\hat{\lambda}(w)|^2\widehat{G}_t(w)\textrm{d}\mu_f(t).
\end{align}

Now we impose the assumption  that $\lambda(\mca)=\lambda(\bbR^d)=0$ so  that $\hat{\lambda}(0)=0$ and the first summation can be taken over all of $\mcv^*$.  Since both of these sums are absolutely convergent, we may combine them to get
\begin{align*}
\int\int K_f^{\rmp}(x,y)\textrm{d}\lambda(x)\textrm{d}\lambda(y)=\sum_{w\in\mcv^*}|\hat{\lambda}(w)|^2\int_0^{\infty}\widehat{G}_t(w)\textrm{d}\mu_f(t).
\end{align*}
Appealing to  equation \eqref{GaussFT},  condition (\ref{gfin}), and Lemma \ref{noz}   completes the proof.
\end{proof}

Now we can prove Theorem \ref{intthrm}.

\begin{proof}[Proof of Theorem \ref{intthrm}]
Suppose $\lambda_1$ and $\lambda_2$ are both minimizers of $\mcl_f$.  The proof of Theorem \ref{posdefker} shows that
\[
\left|\int_{\mca}\int_{\mca}K_f^{\rmp}(x,y)\textrm{d}\lambda_1(x)\textrm{d}\lambda_2(y)\right|<\infty,
\]
so we may apply the parallelogram law to deduce that
\[
\mcl_f\left(\frac{\lambda_1+\lambda_2}{2}\right)=\mcl_f(\lambda_1)-\mcl_f\left(\frac{\lambda_1-\lambda_2}{2}\right)\leq\mcl_f(\lambda_1),
\]
where the inequality is strict unless $\lambda_1=\lambda_2$ (by Theorem \ref{posdefker}).  The minimality of $\lambda_1$ implies $\lambda_1=\lambda_2$ as desired.

In the case $\mathcal{A}=\Omega$, the translation invariance of the periodic problem implies that the unique equilibrium measure $\nu_f$ must be the Haar measure which, restricted to $\Omega$, is Lebesgue measure.
Finally, applying \eqref{lastsum} with $\lambda=\nu_f$ and noting that $\widehat\nu_f(w)=0$ for $w\in \mcv^*\setminus\{0\}$ and $\widehat\nu_f(0)=1$, gives \eqref{ellA}.
\end{proof}

Theorem \ref{intthrm} establishes that the set (\ref{minset}) has a unique element when $f$ satisfies the appropriate hypotheses.  Now we can turn to the proof of Theorem \ref{intasymp}, which we will prove using a standard argument (see \cite[Chapter 2]{Landkof}) that we provide for completeness.

\begin{proof}[Proof of Theorem \ref{intasymp}:]
(I):  Let $\nu_f$ be the unique element of the set (\ref{minset}).  Define
\[
H(x_1,\ldots,x_N):=\sum_{k\neq j}K_f^{\rmp}(x_k,x_j).
\]
Then $\mce_f^{\rmp}(\mca,N)$ is the minimum of $H$ on $\mca^N$.
Therefore, we have the upper bound
\[
\mce_f^{\rmp}(\mca,N)\leq\int_{\mca}\cdots\int_{\mca}H(x_1,\ldots,x_N)\textrm{d}\nu_f(x_1)\cdots \textrm{d}\nu_f(x_N).
\]
This last integral is easily evaluated and equals $N(N-1)$ times the expression on the right hand side of (\ref{n2lim}).  Therefore,
\begin{align}\label{n2upp}
\frac{\mce_f^{\rmp}(\mca,N)}{N(N-1)}\leq\int_{\mca}\int_{\mca}K_f^{\rmp}(x,y)\textrm{d}\nu_f(x)\textrm{d}\nu_f(y).
\end{align}

To get a lower bound, let $K_f^{\ell}(x,y)$ be a kernel that is continuous on $\mca\times\mca$ and satisfies $K_f^{\ell}(x,y)\leq K_f^{\rmp}(x,y)$ for all $x,y\in\mca$.  For each $N\in\bbN$, let $\omega_N\in\mca^N$ be a configuration satisfying $E_f^{\rmp}(\omega_N)=\mce_f^{\rmp}(\mca,N)$.  Let $\nu_N$ be the measure that assigns weight $N^{-1}$ to each point in $\omega_N$ and let $\mcn\subseteq\bbN$ be a subsequence so that $N^{-2}\mce_f^{\rmp}(\mca,N)$ converges to its $\liminf$ as $\bnri$ through $\mcn$.  By taking a further subsequence if necessary, we may assume that $\nu_N$ converges weakly to some probability measure $\nu_{\infty}\in\mcm_{+,1}(\mca)$ as $\bnri$ through $\mcn$.  The continuity of $K_f^{\ell}$ implies
\begin{align*}
\frac{\mce_{f}^{\rmp}(\mca,N)}{N^2}&\geq\int_{\mca}\int_{\mca}K_f^{\ell}(x,y)\textrm{d}\nu_N(x)\textrm{d}\nu_N(y)-\frac{1}{N^2}\sum_{j=1}^NK_f^{\ell}(x_j,x_j)\\
&=\int_{\mca}\int_{\mca}K_f^{\ell}(x,y)\textrm{d}\nu_{\infty}(x)\textrm{d}\nu_{\infty}(y)+o(1)
\end{align*}
as $\bnri$ through $\mcn$.  By taking a supremum over all such continuous $K_f^{\ell}$, we deduce
\[
\liminf_{\bnri}\frac{\mce_{f}^{\rmp}(\mca,N)}{N^2}\geq\int_{\mca}\int_{\mca}K_f^{\rmp}(x,y)\textrm{d}\nu_{\infty}(x)\textrm{d}\nu_{\infty}(y),
\]
where we used Theorem \ref{main}(a) to approximate the kernel $K_f^{\rmp}$ from below by finite, continuous kernels.  In the case that $\mcl_f(\lambda)$ is finite for some $\lambda\in\mcm_{+,1}(\mca)$, we would have a contradiction with (\ref{n2upp}) for large $N$ unless $\nu_{\infty}=\nu_f$, so we have proven the claim.

\vspace{2mm}

(II):  If $\mcl_f(\lambda)=\infty$ for every $\lambda\in\mcm_{+,1}(\mca)$, then our above arguments show that the limit in (\ref{n2lim}) is positive infinity as desired.
\end{proof}

We can also state the following corollary, which was proven in the proof of Theorem \ref{intasymp}:

\begin{corollary}\label{zeros}
Let $f$, $\mu_f$, and $K_f^{\rmp}$ satisfy the hypotheses of Theorem \ref{intasymp}(I) and for each $N\in\bbN$, let $\omega_N$ be a configuration satisfying $E_f^{\rmp}(\omega_N)=\mce_f^{\rmp}(\mca,N)$.  If $\nu_N$ is the measure that assigns weight $N^{-1}$ to each point in $\omega_N$, then $\nu_f$ is the unique weak limit of the measures $\{\nu_N\}_{N\geq2}$ as $\bnri$.
\end{corollary}

We will apply Theorem \ref{intthrm} to some specific examples in the next section, where we discuss potential functions of special interest in more detail.


\section{The Periodic Riesz, log-Riesz, and Logarithmic Potentials}\label{exam}

In this section, we will apply Definition 2 to define the periodic energy associated to some particularly interesting potential functions, namely the Riesz potential, the log-Riesz potential, and the logarithmic potential.  In the case of the Riesz potential, we will also discuss the asymptotic behavior of the minimal energy.

\subsection{The Periodic Riesz Energy.}\label{rsec}  In the section we consider the potential function $f_s(w)=|w|^{-s}$ for any $s>0$.  We will refer to the corresponding energy as the \textit{periodic Riesz} $s$-\textit{energy}.  

First let us briefly describe the situation when $s>d$.  In this case, the sum (\ref{periodf}) converges and has a convenient description in terms of special functions.  Let us denote (as usual) the \textit{Epstein Zeta function} of $\mcv$ by
\[
\zeta_{\mcv}(s):=\sum_{v\in\mcv\setminus\{0\}}|v|^{-s},\qquad s>d,
\]
which is well-defined for $s>d$.  Similarly, we will denote the \textit{Epstein Hurwitz Zeta function} of $\mcv$ by
\[
\zeta_{\mcv}(s;q):=\sum_{v\in\mcv}|q+v|^{-s},\qquad\qquad q\not\in\mcv,\, s>d,
\]
which is also well-defined for $s>d$.  We will see shortly that $\zeta_{\mcv}(s;q)-\frac{2\pi^{d/2}}{\Gamma(\frac{s}{2})(s-d)}$ is actually an entire function of $s\in\bbC$ whenever $q\in\bbR^d\setminus\mcv$.  Now we can write
\[
E_{f_{s}}^{\cmp}((x_j)_{j=1}^N)=\sum_{k\neq j}\zeta_{\mcv}(s;x_{kj}),\qquad\qquad s>d,
\]
with the understanding that $\zeta_{\mcv}(s;v)=\infty$ when $v\in\mcv$.  Properties of the classical periodic Riesz $s$-energy for $s>d$ have been studied before in \cite{CKJAMS,CSIMRN}.

Definition 2 extends the definition of the periodic Riesz $s$-energy to allow for the possibility that $s\leq d$.  For simplicity, we will denote the periodic Riesz $s$-energy by $E_s^{\rmp}$ and the corresponding kernel, and minimal energy by $K_s^{\rmp}$, and $\mce_s^{\rmp}$, respectively.  The kernel that arises from the formula (\ref{1def}) takes the following form:

\begin{theorem}\label{esthm}
The kernel for the periodic Riesz $s$-energy associated with the lattice $\mcv$ is given by
\begin{align}\label{rconc}
K_{s}^{\rmp}(x,y)&=\zeta_{\mcv}(s;x-y)+\frac{2\pi^{d/2}}{\Gamma(\frac{s}{2})(d-s)},\qquad\qquad s>0.
\end{align}
Furthermore, for $0<s<d$,
\begin{align}\label{RieszInt}
 \int_{\Omega}\int_{\Omega} K_{s}^{\rmp}(x,y)\textrm{d}x\textrm{d}y=\frac{2\pi^{d/2}}{\Gamma(\frac{s}{2})(d-s)}.
\end{align}
\end{theorem}

\noindent\textit{Remark.}  An immediate consequence of \eqref{RieszInt} is
\begin{equation}\label{zetaint}
\int_{\Omega}\zeta_{\mcv}(s;x)\textrm{d}x=0, \qquad  0<s<d.
\end{equation}

\begin{proof}
We begin by recalling
\begin{equation}\label{mufs}
y^{-s/2}=\frac{1}{\Gamma(s/2)}\int_0^{\infty}t^{s/2-1}e^{-ty}\textrm{d}t,\qquad y>0,\quad s>0,
\end{equation}
which shows (with $y=|q|^2$) that the Riesz potential is a G-type potential for any $s>0$.  Furthermore, an application of Fubini's Theorem and Morera's Theorem to (\ref{1def}) shows that each term in both sums defining the kernel $K_s^{\rmp}$ is an entire function of $s$.  The uniform convergence of the sums (which follows from the calculations in the proof of Theorem \ref{main}(a)) shows that for any fixed distinct $x,y\in\Omega$, the function $K_s^{\rmp}(x,y)$ is an entire function of $s$.

When $s>d$, we may invoke Theorem \ref{cfthm2} to write
\[
K_s^{\rmp}(x,y)=\lim_{\arz}\left(\sum_{v\in\mcv}|x-y+v|^{-s}e^{-|a(x-y+v)|^2}-\frac{1}{\Gamma(s/2)}\int_0^1\frac{\pi^{d/2}t^{s/2-1}}{(t+a^2)^{d/2}}\textrm{d}t\right).
\]
If we apply Dominated Convergence, then we arrive at the following formula:
\begin{align}\label{ancont}
K_{s}^{\rmp}(x,y)=\zeta_{\mcv}(s;x-y)-\frac{2\pi^{d/2}}{\Gamma(\frac{s}{2})(s-d)},\qquad s>d.
\end{align}
We conclude that $K_s^{\rmp}(x,y)$ provides an analytic continuation of the right-hand side of (\ref{ancont}) to the whole complex plane (see also (\ref{longform}) below).  Since both sides of (\ref{ancont}) are entire functions of $s$ and they are equal on $(d,\infty)$, we must have equality for all $s>0$ as desired.

Finally, by \eqref{ellA} and \eqref{mufs}, we have
$$
 \int_{\Omega}\int_{\Omega} K_{s}^{\rmp}(x,y)\textrm{d}x\textrm{d}y=\frac{\pi^{d/2}}{\Gamma(\frac{s}{2})}\int_1^\infty t^{\frac{s-d}{2}-1}\textrm{d}t=\frac{2\pi^{d/2}}{\Gamma(\frac{s}{2})(d-s)}.
$$
\end{proof}

\noindent\textit{Remark.}  It is worth noting that the fact that the right-hand side of (\ref{ancont}) is an entire function of $s$ also implies that $\zeta_{\mcv}(s,q)-\zeta_{\mcv}(s)$ is an entire function of $s$ (see \cite[page 59]{Terras}).

\vspace{2mm}

Recall the incomplete Gamma function, $\Gamma(\sigma,x)$, given by
\[
\Gamma(\sigma,x):=\int_x^{\infty}t^{\sigma-1}e^{-t}\textrm{d}t.
\]
Evaluating the integrals in the formula (\ref{1def}) yields
\begin{equation}\label{longform}
\begin{split}
K_s^{\rmp}(x,y)=&\frac{\pi^{d/2}}{\Gamma(\frac{s}{2})}\sum_{w\in\mcv^*\setminus\{0\}}e^{2\pi iw\cdot (x-y)}\left(\pi|w|\right)^{s-d}\Gamma\left(\frac{d-s}{2},|\pi w|^2\right)\\
&\qquad\qquad\qquad\qquad\qquad+\frac{1}{\Gamma\left(\frac{s}{2}\right)}\sum_{v\in\mcv}\frac{1}{|x-y+v|^{s}}\Gamma\left(\frac{s}{2},|x-y+v|^2\right)
\end{split}
\end{equation}
(formula (\ref{longform}) was previously known when $d\leq 3$; see \cite[Equation 30]{polytrope}).

This formula enables us to write an explicit expression for the meromorphic continuation of $\zeta_{\mcv}(s;x-y)$ to all of $\bbC$ (see also \cite[Section 10]{Crandall}).  Furthermore, consider the Coulomb case in three dimensions, which corresponds to $s=1$ and $d=3$.  This case is of particular interest because it describes the electrostatic interaction of ions in a three-dimensional crystal.  If we use the identities
\[
\Gamma(1,x)=e^{-x},\qquad\qquad \frac{\Gamma\left(\frac{1}{2},x^2\right)}{\sqrt{\pi}}=\mbox{erfc}(x),
\]
then the right-hand side of (\ref{longform}) becomes
\begin{align}\label{ewaldform}
\sum_{w\in\mcv^*\setminus\{0\}}e^{2\pi iw\cdot(x-y)}\frac{e^{-\pi^2|w|^2}}{\pi|w|^2}+\sum_{v\in\mcv}\frac{\mbox{erfc}(|x-y+v|)}{|x-y+v|},
\end{align}
which is Ewald's formula for the periodic Coulomb potential on a lattice (up to a choice of constants; see \cite[Equation 4]{NaCe}).  Thus we see that Definition 2 enables us to recover this classical result.  

In the non-periodic Riesz energy situation, it is known that if $\mcb\subset\bbR^d$ is a closed $t$-rectifiable set (i.e.; $\mcb$ is the image of a compact set in $\bbR^t$ under a Lipschitz mapping, see \cite{CPD}) and $s>t$, then there is a constant $\textit{C}_{s,t}$ that is independent of $\mcb$ such that
\begin{equation}\label{Cst1}
\lim_{\bnri}\frac{\mce_s(\mcb,N)}{N^{1+s/t}}=\textit{C}_{s,t}\mch_t(\mcb)^{-s/t},\qquad s>t,
\end{equation}
(see \cite{CPD,HSNot,HSAdv}).  When $t=1$, it is known that $\textit{C}_{s,1}=\zeta_{\bbZ}(s)$ (see \cite{MMRS}), however, the exact value of $\textit{C}_{s,t}$ is not known for any values of $s$ or $t$ when $t\geq2$.  It is conjectured that when $t=2$, the constant $\textit{C}_{s,t}$ is equal to the Epstein Zeta function of the equilateral triangular lattice in $\bbR^2$ (see \cite{KS}).  Similar conjectures exist in dimensions $8$ and $24$, where certain canonical lattices are conjectured to resemble the minimal energy configurations for any value of $s>t$ (see \cite[Conjecture 2]{BHS2}).  Indeed, it is these conjectures that motivate the special interest in the Riesz potential.  Our first result establishes a connection between the periodic and non-periodic Riesz energy problems.

\begin{theorem}\label{equalproblem}
Suppose $\mcb\subseteq\overline{\Omega}$ is a compact and $t$-rectifiable set and $s\geq t$, where if $s=t$, we further assume that $\mcb$ is a subset of a $t$-dimensional $C^1$-manifold.  If $0<\mch_t(\mcb)<\infty$ and $\mch_t(\mcb\cap(\overline{\Omega}\setminus\Omega))=0$, then
\begin{align}\label{sameconc}
\lim_{\bnri}\frac{\mce_{s}^{\rmp}(\mcb,N)}{\mce_s(\mcb,N)}=1.
\end{align}
In particular, it holds that
\begin{align}
\label{sgeqt}\lim_{\bnri}\frac{\mce_s^{\rmp}(\mcb,N)}{N^{1+s/t}}&=\frac{\textit{C}_{s,t}}{\mch_t(\mcb)^{s/t}},\qquad s>t,\\
\label{seqt}\lim_{\bnri}\frac{\mce_t^{\rmp}(\mcb,N)}{N^2\log(N)}&=\frac{2\pi^{t/2}}{t\,\Gamma\left(\frac{t}{2}\right)\mch_t(\mcb)}.
\end{align}
\end{theorem}

\noindent\textit{Remark.}  One potential use of this result is that it provides an additional path for deducing the value of the constant $\textit{C}_{s,t}$ mentioned above by studying the minimal energy problem in the periodic setting when $\mcb=\overline{\Omega}$.  See Subsection~\ref{conj} for further details.   In Section \ref{d1}, we use our calculations to again  verify that $\textit{C}_{s,1}=\zeta_{\bbZ}(s)$.

\vspace{2mm}

\noindent\textit{Remark.}  If $s<t$, then the leading order behavior of $\mce_s^{\rmp}(\mcb,N)$ is given by Theorem \ref{intasymp} (see Corollary \ref{slessdexam} below).

\vspace{2mm}

\noindent\textit{Remark.}  The assumption $\mch_t(\mcb\cap(\overline{\Omega}\setminus\Omega))=0$ is not a severe one and we will discuss its implications following the proof of Theorem \ref{equalproblem}.

\begin{proof}[Proof of Theorem \ref{equalproblem}]
The results of \cite{CPD,HSNot,HSAdv} imply that (\ref{sgeqt}) and (\ref{seqt}) hold with $\mce_s^{\rmp}$ replaced by $\mce_s$, so those conclusions will follow immediately once we establish (\ref{sameconc}).  Also, the proof of Theorem \ref{main}(a) shows that the sum
\[
\sum_{w\in\mcv^*\setminus\{0\}}e^{2\pi iw\cdot x}\int_0^1\frac{\pi^{d/2}}{t^{d/2}}e^{-\pi^2|w|^2/t}\textrm{d}\mu_f(t)
\]
is bounded by a constant that is independent of $x\in\bbR^d$.  Therefore, when we sum over all pairs $(k,j)$ with $k\neq j$, the contribution to the energy from this sum is at most a constant multiple of $N(N-1)$, which is negligible for large $N$ compared to $N^{2}\log(N)$.  Therefore, it suffices to consider only the sum over $\mcv$ in (\ref{1def}).

To this end, we resort to equation (\ref{longform}).  It is trivial to see that the entire sum over $\mcv$ is greater than the single term corresponding to $v=0$.  Therefore, by (\ref{longform}) we have
\[
\sum_{k\neq j}\sum_{v\in\mcv}\int_1^{\infty}e^{|x_{kj}+v|^2t}\textrm{d}\mu_f(t)>\frac{1}{\Gamma(s/2)}\sum_{k\neq j}\frac{\Gamma(s/2;|x_{kj}|^2)}{|x_{kj}|^s}.
\]
If we define a kernel $K^*$ on $\mcb$ by
\begin{align}\label{kstar}
K^*(x,y):=\frac{\Gamma(s/2;|x-y|^2)}{\Gamma(s/2)|x-y|^s},
\end{align}
then \cite[Theorems 2 $\&$ 3]{CPD} tell us that the corresponding minimal energy is asymptotically the same as $\mce_s(\mcb,N)$ to leading order as $\bnri$.  Therefore,
\[
\liminf_{\bnri}\frac{\mce_s^{\rmp}(\mcb,N)}{\mce_s(\mcb,N)}\geq1,\qquad\qquad s\geq t.
\]

To bound the $\limsup$, we choose $\delta\in(0,1)$ and define
\[
\delta\mcb:=\mcb\bigcap\left\{x=\sum_{j=1}^da_jv_j:a_j\in[0,\delta],\, j=1,\ldots,d\right\}.
\]
Let us also assume that $\delta$ is large enough so that $\mch_t(\delta\mcb)>0$.  Suppose $\omega_N^{\sharp}\in(\delta\mcb)^N$ is a non-periodic energy minimizing configuration.  Then
\begin{align*}
\mce_s^{\rmp}(\mcb,N)\leq E_s^{\rmp}(\omega_N^{\sharp})&\leq\sum_{{x,y\in\omega_N^{\sharp}}\atop{x\neq y}}|x-y|^{-s}+\sum_{{x,y\in\omega_N^{\sharp}}\atop{x\neq y}}\sum_{v\in\mcv\setminus\{0\}}\frac{\Gamma(s/2;|x-y+v|^2)}{\Gamma(s/2)|x-y+v|^s}\\
&=\mce_s(\delta\mcb,N)+\sum_{{x,y\in\omega_N^{\sharp}}\atop{x\neq y}}\sum_{v\in\mcv\setminus\{0\}}\frac{\Gamma(s/2;|x-y+v|^2)}{\Gamma(s/2)|x-y+v|^s}.
\end{align*}
Since $\delta<1$, this last infinite sum is uniformly bounded in $x,y\in\delta\mcb$, so we can bound the above expression by $\mce_s(\delta\mcb,N)+\mco(N^2)$.  Since $\mch_t(\delta\mcb)>0$, we know that $\mce_s(\delta\mcb,N)$ grows at least as fast as $N^2\log(N)$ as $\bnri$.  Therefore,
\[
\frac{\mce_s^{\rmp}(\mcb,N)}{\mce_s(\delta\mcb,N)}\leq1+o(1),
\]
as $\bnri$.  To complete the proof, we invoke \cite[Theorems A $\&$ B]{CPD} to see that
\[
\lim_{\bnri}\frac{\mce_s(\delta\mcb,N)}{\mce_s(\mcb,N)}=\left(\frac{\mch_t(\mcb)}{\mch_t(\delta\mcb)}\right)^{s/t},\qquad s\geq t.
\]
Since $\delta\in(0,1)$ can be taken arbitrarily close to $1$ and we are assuming $\mch_t(\mcb\cap(\overline{\Omega}\setminus\Omega))=0$, this is the desired result.
\end{proof}

The assumption $\mch_t(\mcb\cap(\overline{\Omega}\setminus\Omega))=0$ in Theorem \ref{equalproblem} prevents the double counting of portions of $\overline{\Omega}$ that differ by an element of $\mcv$.  Indeed, if $\mcv$ is the square lattice in $\bbR^2$ and $\mcb=\partial\Omega$, then it is straightforward to check that $\mce_s^{\rmp}(\mcb,N)=\mce_s^{\rmp}(\partial\Omega\cap\Omega,N)$ and hence
\[
\lim_{\bnri}\frac{\mce_s^{\rmp}(\mcb,N)}{\mce_s(\partial\Omega\cap\Omega,N)}=1,
\]
which shows that Theorem \ref{equalproblem} fails without the assumption $\mch_t(\mcb\cap(\overline{\Omega}\setminus\Omega))=0$.  In some sense, we want $\mcb\subseteq\Omega$, yet we also want $\mcb$ to be compact.  It is not possible to insist on both of these requirements, especially since $\mcb=\overline{\Omega}$ is a meaningful example.  Our assumption implies that ``most" of $\mcb$ is contained in $\Omega$ in an appropriate sense.

If we combine Theorem \ref{equalproblem} with \cite[Theorems 2 $\&$ 3]{CPD}, we deduce the following corollary:

\begin{corollary}\label{equidist}
Suppose $\mcb$ is as in Theorem \ref{equalproblem} and let $\{\omega_N\}_{N\geq2}$ be a sequence of configurations so that $\omega_N\in\mcb^N$ and $E_{s}^{\rmp}(\omega_N)=\mce_{s}^{\rmp}(\mcb,N)$.  If $s\geq t$, then in the weak-$*$ topology, it holds that
\[
\lim_{\nri}\frac{1}{N}\sum_{x\in\omega_N}\delta_x=\frac{\mch_t(\cdot\cap\mcb)}{\mch_t(\mcb)}.
\]
\end{corollary}


\begin{proof}
Let $E_s^*$ be the energy functional associated to the kernel (\ref{kstar}) and let $\mce_s^*$ denote the corresponding minimal energy.  The proof of Theorem \ref{equalproblem} shows that $E_{s}^{\rmp}(\omega_N)>E_s^*(\omega_N)+N(N-1)\alpha_s$ for some $s$-dependent constant $\alpha_s$.  Theorem \ref{equalproblem} and \cite[Theorems 2 $\&$ 3]{CPD} imply that $E_s^*(\omega_N)=\mce_s^*(\mcb,N)(1+o(1))$ as $\bnri$.  The desired conclusion now follows from \cite[Theorems 2 $\&$ 3]{CPD}.
\end{proof}

We can also state a result related to Theorem \ref{equalproblem} that requires fewer geometric assumptions on the set $\mcb$, but assumes a certain separation between translates of $\mcb$ by elements of the lattice.

\begin{theorem}\label{septhrm}
Let $\mcb\subseteq\Omega$ be infinite and compact in $\bbR^d$ and suppose $s>0$ satisfies
\[
\lim_{\bnri}\frac{\mce_s(\mcb,N)}{N^2}=\infty.
\]
Then (\ref{sameconc}) is true.
\end{theorem}

\begin{proof}
The compactness assumption on $\mcb$ assures us that $\mcb=\delta\mcb$ for some $\delta<1$.  Let $K^*$ be given by (\ref{kstar}) and let $\mce_s^*$ be the corresponding minimal energy.  The proof of Theorem \ref{equalproblem} shows that
\[
\mce_s^*(\mcb,N)+\mco(N^2)\leq\mce_s^{\rmp}(\mcb,N)\leq\mce_s(\mcb,N)+\mco(N^2).
\]
Therefore, we need to show that
\[
\liminf_{\bnri}\frac{\mce_s^*(\mcb,N)}{\mce_s(\mcb,N)}\geq1.
\]
Consider an $N$-tuple   $\omega_N=(x_j)_{j=1}^N$  of distinct points in $\mcb$, fix some $i\in\{1,\ldots,N\}$, and choose any $\delta>0$.  We can write
\[
\sum_{{1\leq j\leq N}\atop{j\neq i}}K^*(x_i,x_j)=\sum_{\{j:0<|x_j-x_i|<\delta\}}K^*(x_i,x_j)+\sum_{\{j:|x_j-x_i|\geq\delta\}}K^*(x_i,x_j).
\]
The second of these sums is $\mco(N)$, where the implied constant depends on $\delta$.  Suppose $\epsilon>0$ is given.  The continuity of the incomplete Gamma function in the second argument implies that if $\delta$ is small enough, then
\[
\sum_{\{j:0<|x_j-x_i|<\delta\}}K^*(x_i,x_j)>\sum_{\{j:0<|x_j-x_i|<\delta\}}\frac{1-\epsilon}{|x_i-x_j|^{s}}.
\]
Therefore,
\begin{align*}
\sum_{{1\leq j\leq N}\atop{j\neq i}}K^*(x_i,x_j)&>\sum_{\{j:0<|x_j-x_i|<\delta\}}\frac{1-\epsilon}{|x_i-x_j|^{s}}+\mco(N)=\sum_{j\neq i}\frac{1-\epsilon}{|x_i-x_j|^{s}}+\mco(N).
\end{align*}
If we now sum this relation over $i$, we get
\[
\sum_{{1\leq i,j\leq n}\atop{i\neq j}}K^*(x_i,x_j)>(1-\epsilon)\sum_{{1\leq i,j\leq n}\atop{i\neq j}}|x_i-x_j|^{-s}+\mco(N^2).
\]
Taking the infimum of the left hand side over all $\omega_N\in\mcb^N$ shows
\[
(1-\epsilon)\mce_s(\mcb,N)<\mce_s^*(\mcb,N)+\mco(N^2).
\]
Dividing through by $\mce_s(\mcb,N)$, taking $\bnri$, and then taking $\epsilon\rightarrow0$ completes the proof.
\end{proof}

If $0<s<d$, the form of the kernel given in equation (\ref{longform}) shows it is integrable with respect to Lebesgue measure on the set $\Omega$, so using the last assertion of Theorem~\ref{intthrm} together with Theorems \ref{intasymp} and \ref{esthm},  we  deduce the leading order term in $\mce_s^{\rmp}(\Omega,N)$ as $\bnri$.

\begin{corollary}\label{slessdexam}
If $0<s<d$, then
\begin{align}\label{slessdcase}
\lim_{\bnri}\frac{\mce_s^{\rmp}(\Omega,N)}{N^2}=\int_{\Omega}\int_{\Omega}K_s^{\rmp}(x,y)\,{\rm d}x\,{\rm d}y=\frac{2\pi^{d/2}}{\Gamma(\frac{s}{2})(d-s)}.
\end{align}

\end{corollary}

\medskip

\subsection{Conjectures for Optimal Periodic Riesz $s$-Energy\label{conj}}

Regarding the constant $C_{s,t}$ for $s>t$ appearing in \eqref{Cst1} and \eqref{sgeqt}, it is known (cf. \cite{BHS2}) that
\begin{equation}\label{Clatt}
C_{s,t}\le \zeta_{\mathcal{V}}(s),
\end{equation}
for any $t$-dimensional lattice $\mcv$ of co-volume $1$.   For dimensions $t=2, 4, 8$, and $24$,  it has been conjectured (cf. \cite{BHS2} and references therein) that equality respectively  holds in \eqref{Clatt} for the   equilateral triangular (hexagonal) lattice, the $D_4$ lattice, the $E_8$ lattice, and the Leech lattice.
These conjectures   in turn lead to the conjectured numerical values for the asymptotic energy  expressions in \eqref{sgeqt}.

Denoting the above lattices by $\mcv_2, \mcv_4, \mcv_8$, and $\mcv_{24}$,
we further conjecture that, for all $s>0$, optimal configurations $\omega_N^*$ for the periodic Riesz $s$-energy when ${\mathcal{B}}$ equals the fundamental domain $\Omega=\Omega_t$ for $\mcv_t$ and  $N=m^t$,   $m=2, 3, 4, \ldots$, are given by scaled versions of the lattices restricted to $\Omega$; that is,   $\omega_N^*:=(1/m)\mcv_t\cap \Omega$.   Note that verification  of the optimality of these configurations would confirm the formulas conjectured above for $C_{s,t}$  for $s>t$ and $t=2, 4, 8, 24$.
For $0<s<t$, such optimality would further imply that the following asymptotic formula holds
\begin{equation}
\mce_s^{\rmp}(\Omega_t, N)= \mcl_s^{\rmp} N^2 + \zeta_{\mcv_t}(s)N^{1+s/t} +o(N^{1+s/t}), \qquad (N\to \infty),
\end{equation}
where $ \mcl_s^{\rmp}$ denotes the constant on the right-hand side of \eqref{slessdcase}. (Compare with Conjecture 2 of \cite{KS} and Conjecture 3 of \cite{BHS2} for the non-periodic case of the sphere.)

\subsection{The Periodic Log-Riesz Energy.}\label{logrieszgen}   The log-Riesz $s$-potential is given by $f(x):=|x|^{-s}\log(|x|^{-2})$ for some $s>0$.  The formula (see \cite[page 26]{LapTab})
\[
\frac{-\log(y)}{y^{s/2}}=\frac{1}{\Gamma(s/2)}\int_0^{\infty}t^{s/2-1}(\log(t)-\psi(s/2))e^{-yt}\textrm{d}t,\qquad s,y>0,
\]
(where $\psi$ is the digamma function) shows that the log-Riesz $s$-potential is indeed a G-type potential (with $y=|q|^2$).  One can verify - as in the proof of Theorem \ref{esthm} - that the corresponding periodic kernel $K_{\textrm{log-Riesz},s}^{\rmp}(x,y)$ is analytic as a function of $s\in\{z:\mbox{Re}[z]>0\}$ for any fixed $x$ and $y$ satisfying $x-y\not\in\mcv$.

In the non-periodic setting, the log-Riesz kernel is the derivative of the Riesz kernel with respect to the parameter $s$ and we can extend this notion to the periodic situation.  Indeed, if $s>d$, then we may invoke Theorem \ref{cfthm2} as in the proof of Theorem \ref{esthm} to write (where the prime denotes the derivative with respect to the variable $s$)
\begin{align*}
K_{\textrm{log-Riesz},s}^{\rmp}(x,y)&=\sum_{v\in\mcv}\frac{-2\log(|x-y+v|)}{|x-y+v|^{s}}-\frac{1}{\Gamma(\frac{s}{2})}\int_0^1\frac{\pi^{d/2}(\log(t)-\psi(\frac{s}{2}))}{t^{d/2-s/2+1}}\textrm{d}t\\
&=2\zeta_{\mcv}'(s;x-y)-\frac{1}{\Gamma(\frac{s}{2})}\int_0^1\frac{\pi^{d/2}(\log(t)-\psi(\frac{s}{2}))}{t^{d/2-s/2+1}}\textrm{d}t,\\
&=2\zeta_{\mcv}'(s;x-y)+\frac{2\psi(s/2)\pi^{d/2}}{\Gamma(s/2)(s-d)}+\frac{4\pi^{d/2}}{\Gamma(s/2)(s-d)^2}\\
&=2\frac{d}{ds}K_{s}^{\rmp}(x,y)\qquad\qquad\qquad\qquad\qquad\qquad\qquad\qquad\qquad s>d.
\end{align*}
Since both sides of this equality are analytic functions on the domain $\{z:\mbox{Re}[z]>0\}$, we have proven the following result:

\begin{theorem}\label{eslogthm}
The kernel for the periodic log-Riesz $s$-energy is given by,
\begin{align}\label{logsconc}
\normalfont K_{\textrm{log-Riesz},s}^{\rmp}(x,y)&=2\frac{d}{ds}K_{s}^{\rmp}(x,y)\qquad\qquad s>0.
\end{align}
\end{theorem}

\subsection{The Periodic Logarithmic Energy}\label{logen}

The logarithmic potential is given by $f(x)=\log|x|^{-2}$.  This potential is especially important when $d=2$, where it represents the Coulomb interaction.  Previous attempts have been made to define the logarithmic energy in two dimension (see \cite{log2}), but our result is more general and arises from the same methods used for G-type potentials.  In the non-periodic setting, the logarithmic interaction can be realized as a limiting case of the log-Riesz interaction as the parameter $s$ tends to zero.  We will extend this notion to the periodic setting.

Recall that the logarithmic potential is a weak G-type potential (see (\ref{logiden}) above).  Therefore, equation (\ref{1def}) implies that the corresponding kernel is
\begin{align*}
\nonumber K_{\log}^{\rmp}(x,y)&=\sum_{v\in\mcv}\int_1^{\infty}\frac{e^{-t|x-y+v|^2}}{t}\textrm{d}t+\sum_{w\in\mcv^*\setminus\{0\}}e^{2\pi iw\cdot(x-y)}\int_0^1\frac{\pi^{d/2}}{t^{1+d/2}}e^{-\pi^2|w|^2/t}\textrm{d}t\\
&=\lim_{s\rightarrow0^+}\Gamma\left(\frac{s}{2}\right)\left(\sum_{v\in\mcv}\int_1^{\infty}\frac{e^{-t|x-y+v|^2}}{\Gamma(\frac{s}{2})t^{1-s/2}}\textrm{d}t+\sum_{w\in\mcv^*\setminus\{0\}}e^{2\pi iw\cdot(x-y)}\int_0^1\frac{\pi^{d/2}e^{-\pi^2|w|^2/t}}{\Gamma(\frac{s}{2})t^{1+d/2-s/2}}\textrm{d}t\right)\\
&=\lim_{s\rightarrow0^+}\Gamma\left(\frac{s}{2}\right)K_s^{\rmp}(x,y)\\
&=2\left(\frac{d}{ds}K_{s}^{\rmp}(x,y)\right)\bigg|_{s=0}\\
&=\lim_{s\rightarrow0^+}K_{\textrm{log-Riesz},s}^{\rmp}(x,y)
\end{align*}
where we used the fact that $K_0^{\rmp}(x,y)$ is identically $0$ by (\ref{longform}).  We formally state this conclusion in the following theorem.

\begin{theorem}\label{logforma}
The kernel for the periodic logarithmic energy is given by,
\begin{align}\label{logconc}
\normalfont K_{\textrm{log}}^{\rmp}(x,y)&=\normalfont \lim_{s\rightarrow0^+}K_{\textrm{log-Riesz},s}^{\rmp}(x,y).
\end{align}
\end{theorem}

\vspace{2mm}

One can similarly define the periodic energy for many other potentials by considering the Laplace transform formulas in \cite[pages 24-26]{LapTab}.  We will investigate the minimal periodic energy associated to the Riesz, log-Riesz, and logarithmic kernels in one dimension in Section \ref{d1}.

\section{Convergence Factors and Renormalization}\label{cf2}

In this section, we will revisit and generalize the computational methods used to derive the expression (\ref{1def}).  While never explicitly stated, the formula in Definition 2 is related to other formulas used to sum divergent and conditionally convergent series (see \cite{MP,polytrope,PPdL}).  The method that we will use to derive the formula (\ref{1def}) is that of a \textit{convergence factor} (as in \cite{LPS,MP,polytrope,PPdL}), which is a family of functions $\{g_a\}_{a>0}$ parametrized by the positive real numbers.

If a convergence factor $\{g_a\}_{a>0}$ is given, then for any particular $g_a$ let us define
\begin{align}\label{efg1}
\tilde{E}_{f}^{\rmp}((x_j)_{j=1}^N;g_a):=\sum_{k\neq j}\left(\sum_{v\in\mcv}f(x_k-x_j+v)g_a(x_k-x_j+v)\right).
\end{align}
We will assume that our convergence factors are such that the sum (\ref{efg1}) converges absolutely for all $a>0$.  We will also assume that $\lim_{\arz}g_a(w)=1$ for all $w\in\bbR^d\setminus\{0\}$ and realize our energy functional as a renormalized limit of expressions of the form (\ref{efg1}) as $\arz$.  

Our requirement that $\lim_{\arz}g_a(w)=1$ for all $w\in\bbR^d\setminus\{0\}$ implies that if (\ref{periodf}) is infinite, then as $\arz$, the sum (\ref{efg1}) may tend to infinity.  We will see that in many cases - indeed in all the cases we consider - the sum (\ref{efg1}) can be rewritten as $A_1(a)+A_2(a;(x_j)_{j=1}^N)$, where $A_2(a;(x_j)_{j=1}^N)$ approaches a finite limit as $\arz$ for any (non-degenerate) configuration $(x_j)_{j=1}^N$ and $A_1(a)$ is independent of the configuration.  By writing the sum (\ref{efg1}) in this way, we see that the configurations that minimize $\tilde{E}_{f}^{\rmp}(\cdot;g_a)$ are really minimizing $A_2(a;\cdot)$.  Since we are interested in minimal energy configurations and $A_2(a;\cdot)$ approaches a limit as $\arz$ (call it $A_2(0;\cdot)$), we will define our energy as $A_2(0;\cdot)$.  We will use the Laplace transform and Poisson summation to identify the quantity $A_1(a)$ that we must subtract off of the sum (\ref{efg1}) in order to renormalize it to get a finite limit as $\arz$.  This kind of renormalization procedure has been used previously by applied scientists; indeed the process of renormalizing the Coulomb interaction by subtracting off the quantity $A_1(a)$ is described in \cite{Heyes} as neutralizing each cell in the lattice with a uniform ``background charge." See also \cite{Pay}.

The procedure just outlined begs the question of the dependence of $A_2(0;\cdot)$ on the convergence factor $\{g_a\}_{a>0}$.  We will show that if the convergence factor satisfies some very reasonable smoothness and decay conditions, then the limit $A_2(0;\cdot)$ does not depend on the convergence factor used and is given by the formula (\ref{1def}).  More precisely, we will derive (\ref{1def}) using convergence factors that are Laplace transforms of positive measures.  This generalizes the methods of \cite{polytrope,PPdL}, where only Gaussian convergence factors and G-type potentials are considered.

We will divide our calculations into two parts.  The first will consider the case in which $f$ is a G-type potential.  In fact, we will consider potential functions $f$ and convergence factors $\{g_a\}_{a>0}$ that satisfy the following conditions:
\begin{enumerate}
\item[(CF1)] $f$ is a G-type potential,
\item[(CF2)] for each $a>0$, $g_a(z)$ is finite for all $z\in\bbR^d\setminus\{0\}$ and can be expressed as
\[
g_a(z):=\int_0^{\infty}e^{-|z|^2t}\textrm{d}\mu_{g_a}(t),
\]
for some positive measure $\mu_{g_a}$ on $(0,\infty)$,
\item[(CF3)]  if
\[
f^{\pm}(z):=\int_0^{\infty}e^{-|z|^2t}\textrm{d}\mu_f^{\pm}(t),
\]
then for every $a>0$, the series $\sum_{v\in\mcv}f^{\pm}(q+v)g_a(q+v)$ both converge absolutely for all $q\not\in\mcv$,
\item[(CF4)] $\lim_{\arz}g_a(x)=1$ for all $x\in\bbR^d\setminus\{0\}$.
\end{enumerate}



\vspace{2mm}

For a lattice $\mcv$ generated by $V$ satisfying $\det(V)=1$ and a potential-convergence factor pair $(f,\{g_a\}_{a>0})$ satisfying (CF1-CF4) above, let us define
\begin{align}\label{ef2def}
\normalfont E_{f,1}^{\rmp}((x_j)_{j=1}^N;\{g_a\}_{a>0}):=\limsup_{\arz}\sum_{k\neq j}\left(\sum_{v\in\mcv}f(x_{kj}+v)g_a(x_{kj}+v)-\int_0^1\frac{\pi^{d/2}}{t^{d/2}}\textrm{d}(\mu_f*\mu_{g_a})(t)\right).
\end{align}
To make sure this is well-defined, we must show that the last integral in (\ref{ef2def}) is finite for every $a>0$.  For this, it suffices to consider the case in which $\mu_f$ is positive; the general case follows by considering the positive and negative parts of $\mu_f$ separately.  It is easy to see that
\[
f(z)g_a(z)=\int_0^{\infty}e^{-|z|^2t}\textrm{d}\mu_{a,f}(t)
\]
where $\mu_{a,f}:=\mu_f*\mu_{g_a}$, so the condition (CF3) implies that if $x\not\in\mcv$, then
\[
\sum_{v\in\mcv}\int_0^1e^{-|x+v|^2t}\textrm{d}\mu_{a,f}(t)<\infty.
\]
Since the integrand is positive, we may bring the sum inside the integral and apply Poisson summation to get
\[
\int_0^1\frac{\pi^{d/2}}{t^{d/2}}\sum_{w\in\mcv^*\setminus\{0\}}e^{-\pi^2|w|^2/t}e^{2\pi iw\cdot x}\textrm{d}\mu_{a,f}(t)+\int_0^1\frac{\pi^{d/2}}{t^{d/2}}\textrm{d}\mu_{a,f}(t).
\]
The proof of Theorem \ref{main}(a) shows that the infinite sum converges to an integrable function, so the second integral must also be finite, which is what we wanted to show.

\vspace{2mm}

Our main result for G-type potentials is the following theorem, which shows that this method produces an energy functional that coincides with (\ref{e1def}):

\begin{theorem}\label{cfthm2}
If $f$ and $\{g_a\}_{a>0}$ satisfy conditions (CF1-CF4) stated above and $(x_j)_{j=1}^N$ is a non-degenerate\footnote{Recall this means that $x_i-x_j\not\in\mcv$ for all $i\neq j$.} configuration, then the $\limsup$ in (\ref{ef2def}) is a limit and
\begin{align*}
&E_{f,1}^{\rmp}\left((x_j)_{j=1}^N;\{g_a\}_{a>0})\right)=\sum_{k\neq j}K_f^{\rmp}(x_k,x_j).
\end{align*}
\end{theorem}

The proof will require the following lemma:

\begin{lemma}\label{munot}
Let $\{g_a\}_{a>0}$ satisfy conditions (CF2) and (CF4) in the above list.  For any $0<\epsilon<M<\infty$, the following conclusions hold:
\begin{itemize}[before=\itshape,font=\normalfont]
\item[I)] $\lim_{\arz}\mu_{g_a}([\epsilon,M])=0$.
\item[II)] $\lim_{\arz}\mu_{g_a}((0,\epsilon))=1$.
\end{itemize}
\end{lemma}

\begin{proof}
Define
\[
G_a(x):=\int_0^{\infty}e^{-xt}\textrm{d}\mu_{g_a}(t),\qquad x>0.
\]
To prove part (I), we calculate (for $0<s<r$):
\begin{align*}
0&=\lim_{\arz}\left(G_a(s)-G_a(r)\right)=\lim_{\arz}\int_0^{\infty}(e^{-st}-e^{-rt})\textrm{d}\mu_{g_a}(t)\\&\geq\limsup_{\arz}\int_{\epsilon}^{M+\epsilon}(e^{-st}-e^{-rt})\textrm{d}\mu_{g_a}(t)
\geq\left[\min_{\tau\in[\epsilon,M+\epsilon]}(e^{-s\tau}-e^{-r\tau})\right]\left[\limsup_{\arz}\mu_{g_a}([\epsilon,M+\epsilon))\right],
\end{align*}
which proves the claim.

To prove part (II), we choose some $P>0$ very large and keep the notation that $0<s<r$ and notice
\begin{align*}
1&=\lim_{\arz}\int_0^{\infty}e^{-rt}\textrm{d}\mu_{g_a}(t)\\
&=\lim_{\arz}\left(\int_0^{\epsilon}e^{-rt}\textrm{d}\mu_{g_a}(t)+\int_{\epsilon}^Pe^{-rt}\textrm{d}\mu_{g_a}(t)+\int_P^{\infty}e^{-st}e^{(s-r)t}\textrm{d}\mu_{g_a}(t)\right)\\
&\leq\liminf_{\arz}\left(\mu_{g_a}((0,\epsilon))+e^{-\epsilon r}\mu_{g_a}([\epsilon,P])+e^{(s-r)P}\int_P^{\infty}e^{-st}\textrm{d}\mu_{g_a}(t)\right).
\end{align*}
By part (I), the second term converges to $0$ as $\arz$, and by choosing $P$ large, we can make the last term as small as we want (using the fact that $G_a(s)$ converges to $1$ as $\arz$).  This shows
\[
\liminf_{\arz}\mu_{g_a}((0,\epsilon))\geq1.
\]
To bound the $\limsup$, we calculate
\begin{align*}
1&=\lim_{\arz}\int_0^{\infty}e^{-rt}\textrm{d}\mu_{g_a}(t)\geq\limsup_{\arz}e^{-\epsilon r}\mu_{g_a}((0,\epsilon)),
\end{align*}
which means
\[
\limsup_{\arz}\mu_{g_a}((0,\epsilon))\leq e^{\epsilon r},
\]
for any $r>0$.  Letting $r$ tend to zero proves the result.
\end{proof}

Now we are ready to prove our result for G-type potentials.  The main idea is to establish convergence as $\arz$ of the measures $\mu_{g_a}$ to $\delta_0$ in an appropriate weak sense.

\begin{proof}[Proof of Theorem \ref{cfthm2}]
We will split the proof into two cases.

\vspace{1mm}

\noindent\underline{Case 1:} $\mu_f^-=0$. 

\vspace{1mm}

We write
\begin{align*}
\sum_{v\in\mcv}f(x_{kj}+v)g_a(x_{kj}+v)-\int_0^1\frac{\pi^{d/2}}{t^{d/2}}\textrm{d}\mu_{a,f}(t)=\sum_{v\in\mcv}\int_0^{\infty}e^{-|x_{kj}+v|^2t}\textrm{d}\mu_{a,f}(t)-\int_0^1\frac{\pi^{d/2}}{t^{d/2}}\textrm{d}\mu_{a,f}(t).
\end{align*}
To evaluate the integrals in the infinite sum, we split them into two integrals, one ranging from $0$ to $1$ and the other from $1$ to infinity.  It follows immediately from the definition of convolution (see \cite[page 270]{Folland}) that
\begin{align}
\nonumber&\int_1^{\infty}e^{-|x_{kj}+v|^2t}\textrm{d}\mu_{a,f}(t)=\int_0^{\infty}\int_0^{\infty}\chi_{\{z+y\geq1\}}(z+y)e^{-|x_{kj}+v|^2(z+y)}\textrm{d}\mu_f(z)\textrm{d}\mu_{g_a}(y)\\
\nonumber&\qquad\qquad=\int_0^{\infty}\left(\int_{(1-y)^+}^{\infty}e^{-|x_{kj}+v|^2z}\textrm{d}\mu_f(z)\right)e^{-|x_{kj}+v|^2y}\textrm{d}\mu_{g_a}(y)\\
\label{ltint}&\qquad\qquad=\int_0^{\epsilon}\left(\int_{(1-y)^+}^{\infty}e^{-|x_{kj}+v|^2z}\textrm{d}\mu_f(z)\right)e^{-|x_{kj}+v|^2y}\textrm{d}\mu_{g_a}(y)\\
\label{ltint2}&\qquad\qquad\qquad\qquad\qquad+\int_\epsilon^{\infty}\left(\int_{(1-y)^+}^{\infty}e^{-|x_{kj}+v|^2z}\textrm{d}\mu_f(z)\right)e^{-|x_{kj}+v|^2y}\textrm{d}\mu_{g_a}(y),
\end{align}
where $\epsilon$ is some small positive number and $(1-y)^+=\max\{1-y,0\}$.  The integral (\ref{ltint}) is easy to understand as $\arz$.  Indeed, Lemma \ref{munot} implies that the restriction of $\mu_{g_a}$ to $[0,\epsilon]$ converges weakly to $\delta_0$ as $\arz$ and the $y$-integrand in (\ref{ltint}) is right-continuous at $0$, so as $\arz$, the integral converges to
\[
\int_1^{\infty}e^{-|x_{kj}+v|^2z}\textrm{d}\mu_f(z).
\]
The integral (\ref{ltint2}) can be bounded above in absolute value by
\begin{align}\label{gtozero}
\left(\int_0^{\infty}e^{-|x_{kj}+v|^2z}\textrm{d}\mu_f(z)\right)\left(\int_{\epsilon}^{\infty}e^{-|x_{kj}+v|^2y}\textrm{d}\mu_{g_a}(y)\right).
\end{align}
The second factor in (\ref{gtozero}) can be rewritten as
\[
g_a(x_{kj}+v)-\int_0^{\epsilon}e^{-|x_{kj}+v|^2y}\textrm{d}\mu_{g_a}(y)
\]
Lemma \ref{munot} implies the restriction of $\mu_{g_a}$ converges weakly to $\delta_0$ as $\arz$, so condition (CF4) implies this expression converges to $0$ as $\arz$.  Therefore the limit of (\ref{gtozero}) as $\arz$ is zero, which implies
\[
\lim_{\arz}\int_1^{\infty}e^{-|x_{kj}+v|^2t}\textrm{d}\mu_{a,f}(t)=\int_1^{\infty}e^{-|x_{kj}+v|^2t}\textrm{d}\mu_f(t),\qquad v\in\mcv.
\]
We may apply Dominated Convergence (using calculations from the proof of Theorem \ref{main}(a)) to conclude that
\begin{align}\label{piece1a}
\lim_{\arz}\sum_{v\in\mcv}\int_1^{\infty}e^{-|x_{kj}+v|^2t}\textrm{d}\mu_{a,f}(t)=\sum_{v\in\mcv}\int_1^{\infty}e^{-|x_{kj}+v|^2t}\textrm{d}\mu_f(t).
\end{align}

It remains to evaluate
\begin{align}\label{01piece}
\lim_{\arz}\left(\sum_{v\in\mcv}\int_0^1e^{-|x_{kj}+v|^2t}\textrm{d}\mu_{a,f}(t)-\int_0^1\frac{\pi^{d/2}}{t^{d/2}}\textrm{d}\mu_{a,f}(t)\right)
\end{align}
We apply Poisson summation:
\[
\sum_{v\in\mcv}e^{-t|x_{kj}+v|^2}=\frac{\pi^{d/2}}{t^{d/2}}\sum_{w\in\mcv^*}e^{-\pi^2|w|^2/t}e^{2\pi iw\cdot x_{kj}},
\]
to rewrite (\ref{01piece}) as
\[
\lim_{\arz}\left(\sum_{w\in\mcv^*\setminus\{0\}}e^{2\pi iw\cdot x_{kj}}\int_0^1\frac{\pi^{d/2}}{t^{d/2}}e^{-\pi^2|w|^2/t}\textrm{d}\mu_{a,f}(t)\right).
\]
For each term in this sum, we rewrite the integral as
\begin{align}
\label{secterm}&\int_0^{\epsilon}\int_0^{1-y}\frac{\pi^{d/2}}{(z+y)^{d/2}}e^{-\pi^2|w|^2/(z+y)}\textrm{d}\mu_f(z)\textrm{d}\mu_{g_a}(y)\\
\nonumber&\qquad\qquad\qquad\qquad\qquad+\int_{\epsilon}^1\int_{0}^{1-y}\frac{\pi^{d/2}}{(z+y)^{d/2}}e^{-\pi^2|w|^2/(z+y)}\textrm{d}\mu_f(z)\textrm{d}\mu_{g_a}(y).
\end{align}
The second term in (\ref{secterm}) is easily bounded above by a constant multiple of $\mu_{g_a}([\epsilon,1])$, which tends to zero as $\arz$ by Lemma \ref{munot}(I).  To calculate the limit as $\arz$ of the first term, we again notice that the $y$-integrand is right continuous at $0$ and Lemma \ref{munot} implies that the restriction of $\mu_{g_a}$ to $[0,\epsilon]$ converges weakly to $\delta_0$ as $\arz$, so the first integral in (\ref{secterm}) converges as $\arz$ to
\[
\int_0^1\left(\frac{\pi}{z}\right)^{d/2}e^{-\pi^2|w|^2/z}\textrm{d}\mu_f(z).
\]
It follows that for each $w\in\mcv^*\setminus\{0\}$ it holds that
\[
\lim_{\arz}e^{2\pi ix_{kj}\cdot w}\int_0^1\frac{\pi^{d/2}}{t^{d/2}}e^{-\pi^2|w|^2/t}\textrm{d}\mu_{a,f}(t)=e^{2\pi ix_{kj}\cdot w}\int_0^1\frac{\pi^{d/2}}{t^{d/2}}e^{-\pi^2|w|^2/t}\textrm{d}\mu_f(t),
\]
and Dominated Convergence again allows us to make the same conclusion for the sum over all $w\in\mcv^*\setminus\{0\}$.

\vspace{2mm}

\noindent\underline{Case 2:}  $\mu_f=\mu_f^+-\mu_f^-$ with $\mu_f^-$ finite.

\vspace{2mm}

Case 1 implies that if we replace $f$ by $f^{\pm}$ and $\mu_f$ by $\mu_f^{\pm}$ in (\ref{ef2def}), then the the conclusion of the theorem is valid.  Therefore, the same must be true if we replace $f$ by $f^+-f^-$ and $\mu_f$ by $\mu_f^+-\mu_f^-$ (condition (CF3) allows us to rearrange the sums).  This is the desired conclusion.  We see that Case 2 is the separate application of Case 1 to the potentials $f^+$ and $f^-$.
\end{proof}

Now we will apply Theorem \ref{cfthm2} to derive (\ref{1def}) for weak G-type potentials.  We will require the convergence factor $\{g_a\}_{a>0}$ to satisfy conditions (CF2) and (CF4) above and we replace condition (CF3) with the stronger requirement that
\begin{align}\label{gasum}
\sum_{v\in\mcv}g_a(u+v)<\infty,\qquad u\in\bbR^d\setminus\mcv,\quad a>0.\tag{CF5}
\end{align}
Recall that $\{v_j\}_{j=1}^d$ are the columns of the matrix $V$ that determine the lattice $\mcv$ and set $u^*:=\frac{1}{2}(v_1+\cdots+v_d)$.  If $f$ is a weak G-type potential, then we define $\mu_{f,\alpha}$ to be the restriction of $\mu_f$ to the interval $[\alpha,\infty)$ and
\begin{align}\label{weakform}
\nonumber&\qquad\qquad\qquad\qquad\qquad\qquad f_{\alpha}(z):=\int_{0}^{\infty}e^{-|z|^2t}\textrm{d}\mu_{f,\alpha}(t),\\
&K_{f,2}^{\rmp}(x,y):=\limsup_{\alrz}\limsup_{\arz}\bigg(\sum_{v\in\mcv}(f_{\alpha}(x-y+v)+f^*(\alpha))g_a(x-y+v)\\
\nonumber&\qquad\qquad\qquad\qquad\qquad\qquad\qquad\qquad-f^*(\alpha)\sum_{v\in\mcv}g_a(u^*+v)-\int_0^1\frac{\pi^{d/2}}{t^{d/2}}\textrm{d}(\mu_{f,\alpha}*\mu_{g_a})(t)\bigg),
\end{align}
and
\[
E_{f,2}^{\rmp}((x_j)_{j=1}^N;\{g_a\}_{a>0}):=\sum_{k\neq j}K_{f,2}^{\rmp}(x_k,x_j).
\]
Our result for weak G-type potentials takes the following form:

\begin{theorem}\label{cfweakthm}
If $f$ is a weak G-type potential with measure $\mu_f$; $\{g_a\}_{a>0}$ satisfies (CF2), (CF4), and (\ref{gasum}); and $(x_j)_{j=1}^N$ is a non-degenerate configuration, then the $\limsup$'s in (\ref{weakform}) are both limits and
\begin{align*}
&E_{f,2}^{\rmp}\left((x_j)_{j=1}^N;\{g_a\}_{a>0})\right)=\sum_{k\neq j}K_f^{\rmp}(x_k,x_j).
\end{align*}
\end{theorem}

As in the case of G-type potentials, the proof will require a technical lemma.

\begin{lemma}\label{someg}
Suppose $\{g_a\}_{a>0}$ is a convergence factor satisfying (CF2), (CF4), and (\ref{gasum}).  Then
\[
\lim_{\arz}\sum_{v\in\mcv}(g_a(u+v)-g_a(u^*+v))=0,\qquad u\in\bbR^d\setminus\mcv.
\]
\end{lemma}

\begin{proof}
We write
\[
\sum_{v\in\mcv}(g_a(u+v)-g_a(u^*+v))=\sum_{v\in\mcv}\int_0^{\infty}e^{-|u+v|^2t}-e^{-|u^*+v|^2t}\textrm{d}\mu_{g_a}(t).
\]
We split the integral at $1$ and notice that the integral from $1$ to infinity converges to $0$ as $\arz$ by the same reasoning that showed that the expression (\ref{gtozero}) converges to $0$ as $\arz$.  To calculate the integral from $0$ to $1$, we bring the sum inside the integral (which is justified by (\ref{gasum})) and apply Poisson summation to rewrite the integral as
\[
\int_0^1\sum_{w\in\mcv^*}(e^{2\pi iw\cdot u}-e^{2\pi iw\cdot u^*})\frac{\pi^{d/2}}{t^{d/2}}e^{-\pi^2|w|^2/t}\textrm{d}\mu_{g_a}(t).
\]
The $w=0$ term contributes $0$ to this sum, while the remaining sum converges for all $t\in[0,1]$ to a continuous function that is $0$ at $0$ (as in the proof of Theorem \ref{main}(a)).  Lemma \ref{munot} implies $\mu_{g_a}$ restricted to $[0,1]$ converges to $\delta_0$ as $\arz$, so this integral converges to $0$ as $\arz$ as desired.
\end{proof}

\begin{proof}[Proof of Theorem \ref{cfweakthm}]
We write
\begin{align*}
&\sum_{v\in\mcv}(f_{\alpha}(x_{kj}+v)+f^*(\alpha))g_a(x_{kj}+v)-f^*(\alpha)\sum_{v\in\mcv}g_a(u^*+v)-\int_0^1\frac{\pi^{d/2}}{t^{d/2}}\textrm{d}(\mu_{f,\alpha}*\mu_{g_a})(t)\\
&\qquad\qquad\qquad=\sum_{v\in\mcv}f_{\alpha}(x_{kj}+v)g_a(x_{kj}+v)-\int_0^1\frac{\pi^{d/2}}{t^{d/2}}\textrm{d}(\mu_{f,\alpha}*\mu_{g_a})(t)\\
&\qquad\qquad\qquad\qquad\qquad\qquad\qquad\qquad-f^*(\alpha)\sum_{v\in\mcv}\left(g_a(u^*+v)-g_a(x_{kj}+v)\right).
\end{align*}
Lemma \ref{someg} implies that the last term in this expression tends to $0$ as $\arz$, while Theorem \ref{cfthm2} implies the rest converges to $K_{f_{\alpha}}^{\rmp}(x_k,x_j)$ as $\arz$.  This kernel has the form (\ref{1def}) but with $\mu_f$ replaced by $\mu_{f,\alpha}$.  If we then take $\alrz$, we recover the desired result.
\end{proof}

It will be beneficial to have some concrete examples to consider to help us understand the above calculations.

\vspace{2mm}

\noindent\textbf{Example: Riesz Convergence Factors.}
This example highlights the fact that for G-type potentials, we do not need the convergence factor to be absolutely summable on its own (as in (CF5)), but only require that the weaker condition (CF3) be satisfied.  Consider the Riesz potential $f(x)=|x|^{-s}$ for $s\geq d$ and the Riesz convergence factor $g_a(x)=|x|^{-a}$.  In this case,
\[
\textrm{d}\mu_{g_a}(t)=\frac{1}{\Gamma(a/2)}t^{a/2-1}\textrm{d}t.
\]
Since $f$ and $g_a$ are both positive, the condition (CF3) reduces to
\[
\sum_{v\in\mcv}f(q+v)g_a(q+v)<\infty,\qquad q\not\in\mcv,
\]
which is true in this case because $s\geq d$.  We have already seen that the Riesz potential is a G-type potential when $s\geq d$, and so this potential and convergence factor satisfy conditions (CF1-CF4) listed above.  Therefore, Theorem \ref{cfthm2} tells us that we will recover (\ref{1def}) as our energy by this method.

\vspace{2mm}

\noindent\textbf{Example: Gaussian Convergence Factors.}
Consider the logarithmic potential $f(x)=-\log|x|^2$ and the Gaussian convergence factor $g_a(x)=e^{-|ax|^2}$, which was utilized in \cite{polytrope}.  In this case,
\[
\textrm{d}\mu_{g_a}(t)=\delta_{a^2},
\]
so it is clear that this choice of convergence factor satisfies the conditions of Lemma \ref{someg}.  We have already seen that the logarithmic potential is a weak G-type potential, so Theorem \ref{cfweakthm} tells us that we will recover (\ref{1def}) as our energy by using this convergence factor.

\vspace{2mm}

We have shown that the formula (\ref{1def}) appears naturally as a definition of the periodic energy for a variety of potentials and results from the natural process of using a convergence factor with the appropriate renormalization.  It may be possible to work out an exact set of hypotheses on the pair $(f,\{g_a\}_{a>0})$ for the resulting energy to coincide with (\ref{e1def}), but that is not our purpose here.  The generality of our current result combined with the nice properties of the energy given by (\ref{e1def}) are sufficient to justify our use of (\ref{e1def}) as a definition of a periodic energy functional.

\begin{proof}[Proof of Theorem \ref{main}(c)]
Theorem \ref{cfthm2} shows that
\[
K_f^{\rmp}(x,y)=\lim_{\arz}\left(\sum_{v\in\mcv}f(x-y+v)e^{-|a(x-y+v)|^2}-\int_0^{1}\frac{\pi^{d/2}}{(t+a^2)^{d/2}}\textrm{d}\mu_f(t)\right).
\]
However, since (\ref{periodf}) is absolutely convergent we can bring the limit inside the sum over $\mcv$.  Therefore,
\[
E_f^{\rmp}((x_j)_{j=1}^N)=\sum_{k\neq j}\left(\sum_{v\in\mcv}f(x_{kj}+v)\right)-N(N-1)\lim_{\arz}\int_0^1\frac{\pi^{d/2}}{(t+a^2)^{d/2}}\textrm{d}\mu_f(t).
\]
Since this last limit - which must be finite in this case - does not depend on the configuration, we have proven the result.
\end{proof}

We will conclude this section with an application of our results to the Laplace transform.  Suppose that $f$ is a G-type potential, and $\{g_a\}_{a>0}$ is the Gaussian convergence factor $g_a(x)=e^{-|ax|^2}$.  The proof of Theorem \ref{main}(c) shows that if the sum (\ref{periodf}) is absolutely convergent, then the following exists and is finite:
\[
\lim_{\arz}\int_0^{1}\frac{1}{(t+a^2)^{d/2}}\textrm{d}\mu_f(t).
\]
In other words, if the potential has sufficiently fast decay at infinity, then its inverse Laplace transform must have a certain minimum amount of decay at $0$.  We will state this conclusion as the following proposition:

\begin{prop}\label{l1sum}
Suppose $F(r):(0,\infty)\rightarrow\bbR$ is given by
\[
\normalfont F(r):=\int_0^{\infty}e^{-rt}\textrm{d}\mu(t)
\]
for some signed measure $\mu$ with $\mu^-$ finite.  Suppose further that $\sum_{v\in\mcv}F\left(|q+v|^2\right)$ converges absolutely for some $q\in\bbR^d\setminus\mcv$.  Then the following exists and is finite:
\[
\normalfont\lim_{\arz}\int_0^{1}\frac{1}{(t+a^2)^{d/2}}\textrm{d}\mu(t).
\]
\end{prop}

\section{The $d=1$ case}\label{d1}

In this section we consider the minimal energy of the unit interval $\Omega=[0,1)$ for the periodic Riesz kernel for all positive values of $s$, the log-Riesz kernel for all positive $s$, and the logarithmic kernel.   Of fundamental importance to our calculations is the following result, which follows from a standard ``winding number argument'' of L. Fejes T\'{o}th (see \cite[Proposition 1]{BHS}).

\begin{prop}\label{thmEqSpaced}
Suppose $K$ is a kernel on $[0,1)\times[0,1)$ of the form $K(x,y)=\phi(|x-y|)$ for
some lower semi-continuous function $\phi:[0,1)\to \bbR\cup\{+\infty\}$ that is (a) strictly convex on $(0,1)$ and (b) satisfies $\phi(x)=\phi(1-x)$ for $x\in(0,1)$.  Then an ordered configuration of $N$ points $0\le x_1^*\le\cdots \le x_N^*<1$ minimizes the $N$-point $K$-energy given by
$$\sum_{{1\leq i,j\leq N}\atop{i\neq j}}K(x_i,x_j),$$
over all $N$ point configurations in $[0,1)$ if and only if there is some $0\le \alpha<1/N$ such that $x_j=\alpha+\frac{j-1}{N}$ for all $j=1,\ldots, N$.
\end{prop}

In all of our examples, we will verify that the kernels satisfy the hypotheses of Proposition \ref{thmEqSpaced} and so deduce the minimal energy configurations.  This will allow us to compute exact formulas for the minimal energy.

\subsection{The Riesz kernel.}

Since we always assume that $\det(V)=1$, we must have $\Omega=[0,1)$.  Next, let us recall the Hurwitz Zeta function
\[
\zeta(s;q)=\sum_{n=0}^{\infty}(q+n)^{-s},\qquad q>0,\,\, s>1.
\]
Recall the form of the periodic Riesz kernel
\begin{align}\label{newkdef}
K_s^{\rmp}(x,y):=\zeta_{\bbZ}(s;x-y)-\frac{2\sqrt{\pi}}{\Gamma(s/2)(s-1)}.
\end{align}
Notice that the Epstein Zeta function for the integer lattice is just twice the Riemann Zeta function $\zeta(s)$.  Therefore, we can use (\ref{newkdef}) to write the energy functional in this setting as
\begin{align}\label{edef}
E_{s}^{\rmp}((x_j)_{j=1}^N)=\sum_{{k,j=1}\atop{k\neq j}}^N\bigg(\zeta(s;|x_{kj}|)+\zeta(s;1-|x_{kj}|)-\frac{2\sqrt{\pi}}{\Gamma(s/2)(s-1)}\bigg),\qquad s\neq1.
\end{align}
The case $s=1$ will require special attention, but we have already seen that the Riesz kernel is an entire function of $s$, so we will be able to make sense of the periodic Riesz 1-energy.

Define the function $J_s(q)=\zeta(s;q)+\zeta(s,1-q)-\frac{2\sqrt{\pi}}{\Gamma(s/2)(s-1)}$.  Notice that $J_s(q)=J_s(1-q)$ and since
\[
\frac{\partial}{\partial q}\zeta(s;q)=-s\zeta(s+1;q),
\]
we have
\begin{align}
\nonumber J_s''(q)&=\frac{\partial^2}{\partial q^2}\left(\zeta(s;q)+\zeta(s;1-q)\right)=s(s+1)\left(\zeta(s+2;q)+\zeta(s+2;1-q)\right)>0,\quad q\in(0,1).
\end{align}
This shows that the function $J_s$ is convex on $(0,1)$ and so Proposition \ref{thmEqSpaced} implies that the energy minimizing configuration is $N$ equally spaced points in the unit interval.  This fact and a simple calculation allow us to write
\begin{align}\label{kerform}
\mce_{s}^{\rmp}([0,1),N)=2N\sum_{j=1}^{N-1}\zeta\left(s;\frac{j}{N}\right)-N(N-1)\frac{2\sqrt{\pi}}{\Gamma(s/2)(s-1)}.
\end{align}
We need the following formula, the proof of which can be deduced from \cite[p. 249]{Apostol}:

\begin{lemma}\label{multform}
For any $n\in\bbN$, it holds that
\begin{align*}
\sum_{j=1}^n\zeta\left(s;\frac{j}{n}\right)=n^s\zeta(s).
\end{align*}
\end{lemma}

By invoking the lemma, we arrive at the following:

\begin{align}\label{sneq1a}
\mce_{s}^{\rmp}([0,1),N)=2N^{1+s}\zeta(s)-2N\zeta(s)-N(N-1)\frac{2\sqrt{\pi}}{\Gamma(s/2)(s-1)},\qquad\qquad s\neq1.
\end{align}

However, we have already seen that the energy minimizing configurations are independent of $s$ and that the energy of a fixed configuration is an analytic (and hence continuous) function of $s$.  Therefore, the formula (\ref{sneq1a}) is also valid when $s=1$.  We have therefore proven the following:

\begin{theorem}\label{snotone}
If $s\in(0,\infty)$, then the minimal periodic Riesz $s$-energy of the unit interval is given by
\begin{align}\label{slessconc}
\mce_{s}^{\rmp}([0,1),N)=
\begin{cases}
2N^{1+s}\zeta(s)-2N\zeta(s)-N(N-1)\frac{2\sqrt{\pi}}{\Gamma(s/2)(s-1)},\qquad& s\neq 1\\
2N^2\log(N)+2N(N-1)\gamma&s=1,
\end{cases}
\end{align}
where $\gamma$ is the Euler-Mascheroni constant.
\end{theorem}

Notice that in the expression (\ref{slessconc}), there are no limits or error terms; we have an exact formula.

\subsection{The Log-Riesz Kernel and the Logarithmic Kernel}\label{lriesz}

As mentioned in Section \ref{logrieszgen}, the log-Riesz kernel is given by the derivative of the Riesz kernel with respect to the parameter $s$.  Consider the kernel given by
\[
R_s(x)=2\frac{\partial}{\partial s}\left(\zeta(s;x)+\zeta(s;1-x)-\frac{2\sqrt{\pi}}{\Gamma\left(\frac{s}{2}\right)(s-1)}\right),\qquad s\in(0,\infty)\setminus\{1\}.
\]
For simplicity, we will presently only consider the case $s\neq1$; we will obtain our results for $s=1$ by continuity as in the previous section.  Our first step is to verify that the minimal energy configuration is equally spaced points in the interval.  We again proceed by a derivative calculation.  Indeed, we have (where $'$ indicates a derivative with respect to $s$)
\begin{align*}
\frac{\partial^2}{\partial q^2}\frac{R_s(q)}{2}&=s(s+1)(\zeta'(s+2;q)+\zeta'(s+2;1-q))+(2s+1)(\zeta(s+2;q)+\zeta(s+2;1-q)).
\end{align*}

It is clear that $\zeta(s+2;q)+\zeta(s+2;1-q)$ is positive, so let us turn our attention to the terms involving derivatives.  Let us write
\[
\zeta(s+2;q)+\zeta(s+2;1-q)=\frac{1}{q^{s+2}}+\frac{1}{(1-q)^{s+2}}+\sum_{n=1}^{\infty}\left(\frac{1}{(q+n)^{s+2}}+\frac{1}{(1-q+n)^{s+2}}\right).
\]
Differentiating either of the first two terms with respect to $s$ will yield a positive result, while differentiating the infinite sum will yield a negative result.  More precisely, we have
\begin{align}\label{q1}
\frac{\partial}{\partial s}\left(\frac{1}{q^{s+2}}+\frac{1}{(1-q)^{s+2}}\right)=\frac{\log(1/q)}{q^{s+2}}+\frac{\log(1/(1-q))}{(1-q)^{s+2}}.
\end{align}
A straightforward calculation reveals that
\begin{align*}
&\frac{\partial^2}{\partial q^2}\left(\frac{\log(1/q)}{q^{s+2}}+\frac{\log(1/(1-q))}{(1-q)^{s+2}}\right)=\\
&\qquad=\frac{5+2s-(s^2+5s+6)\log(q)}{q^{s+4}}+\frac{5+2s-(s^2+5s+6)\log(1-q)}{(1-q)^{s+4}}>0.
\end{align*}
Therefore, the symmetry of the expression (\ref{q1}) implies that the absolute minimum  of (\ref{q1}) is obtained when $q=1/2$, where it takes the value $2^{s+3}\log(2)>8\log(2)\approx5.542$.  Therefore, the positive contribution to the derivative of $\zeta(s+2;q)+\zeta(s+2;1-q)$ is at least this large.

The negative contribution to the derivative can be bounded above in absolute value by
\begin{align*}
\sum_{n=1}^{\infty}\left(\frac{\log(q+n)}{(q+n)^{s+2}}+\frac{\log(1-q+n)}{(1-q+n)^{s+2}}\right)\leq2\sum_{n=1}^{\infty}\frac{\log(n+1)}{n^{s+2}}\leq2\sum_{n=1}^{\infty}\frac{\log(n+1)}{n^2}.
\end{align*}
This last sum is easily evaluated numerically, and it is in fact less than $4$.

If we combine the positive and negative contributions to the derivative of $\zeta(s+2;q)+\zeta(s+2;1-q)$, then we see that $\zeta'(s+2;q)+\zeta'(s+2;1-q)$ is positive for all $q\in(0,1)$.  It follows that the second derivative of $R_s$ is positive for all $q\in(0,1)$.  Therefore, we invoke Proposition \ref{thmEqSpaced} to conclude that the minimal energy configuration is given by equally spaced points in the interval.

Since (\ref{slessconc}) is an exact formula, we can obtain an exact formula for the minimal energy corresponding to the log-Riesz kernel on $[0,1)$ by differentiating both sides of (\ref{slessconc}) with respect to $s$.  Theorem \ref{eslogthm} implies the log-Riesz kernel is continuous as a function of $s$, so we get the desired result for $s=1$ also.

\begin{theorem}\label{logriesz1}
If $s>0$, then the minimal periodic log-Riesz $s$-energy of the unit interval is given by
\begin{align}\label{logsmin}
\normalfont \mce_{\textrm{log-Riesz},s}^{\rmp}([0,1),N)\\
\nonumber&\hspace{-1.17in}=
\begin{cases}
4\left[N^{1+s}\log(N)\zeta(s)+\zeta'(s)N(N^s-1)+\sqrt{\pi}N(N-1)\frac{\Gamma'\left(\frac{s}{2}\right)\frac{s-1}{2}+\Gamma\left(\frac{s}{2}\right)}{\Gamma\left(\frac{s}{2}\right)^2(s-1)^2}\right],\quad &s\neq1\\
2(N\log(N))^2+4\gamma N^2\log(N)-N(N-1)\left(4\gamma_1+\frac{1}{2}(\psi(\frac{1}{2})^2-\psi'(\frac{1}{2}))\right),\qquad &s=1,
\end{cases}
\end{align}
where $\psi(z)$ is the polygamma function, $\gamma$ is the Euler-Mascheroni constant, and
\[
\gamma_1=\lim_{m\rightarrow\infty}\left(-\frac{\log(m)^2}{2}+\sum_{k=1}^m\frac{\log(k)}{k}\right),
\]
is the negative of the coefficient of $(s-1)$ in the Laurent expansion of $\zeta(s)$ around $1$.
\end{theorem}

Since equally spaced points minimize the periodic log-Riesz $s$-energy for all $s>0$, it follows easily from Theorem \ref{logforma} that the same is true of the periodic logarithmic energy.  If we combine this with Theorem \ref{logriesz1}, we get the following:

\begin{theorem}\label{log1}
The minimal periodic logarithmic energy of the unit interval satisfies
\[
\normalfont \mce_{\textrm{log}}^{\rmp}([0,1),N)=2N\left(\sqrt{\pi}(N-1)-\log(N)\right).
\]
\end{theorem}


\section*{Appendix A: Poisson Summation on Bravais Lattices}\label{psl}

Here we will state and prove some of the necessary formulas for Poisson summation.  The methods and ideas here are not new, but in the literature there is widespread inconsistency concerning notation and proper normalization, so some calculation is required for clarity.  For a function $f:\bbR^d\rightarrow\bbR$ that is in $L^1(\bbR^n)$, we define its Fourier transform by
\[
\hat{f}(y)=\int_{\bbR^d}f(t)e^{-2\pi i y\cdot t}\textrm{d}t.
\]
The Poisson summation formula states that if $f$ and $\hat{f}$ have sufficient decay at infinity, then
\begin{align}\label{PF}
\sum_{k\in\bbZ^d}f(k)=\sum_{m\in\bbZ^d}\hat{f}(m)
\end{align}
(see \cite[page 254]{Folland}).

Given a lattice $\mcv$ determined by a matrix $V$ as in our above results, let us fix some $x\in\bbR^d$ and $\omega\in(0,\infty)$ and define
\[
f(z)=e^{-\omega|x+Vz|^2}.
\]
This function $f$ has sufficient decay at infinity to apply the Poisson summation formula, so we have
\[
\sum_{v\in\mcv}e^{-\omega|x+v|^2}=\sum_{k\in\bbZ^d}f(k)=\sum_{m\in\bbZ^d}\hat{f}(m).
\]
Therefore, we need to calculate the Fourier transform of $f$.  We have
\begin{align*}
\hat{f}(y)&=\int_{\bbR^d}e^{-\omega|V(V^{-1}x+t)|^2}e^{-2\pi i y\cdot t}\textrm{d}t\\
&=e^{2\pi iy\cdot V^{-1}x}\int_{\bbR^d}e^{-\omega|Vu|^2}e^{-2\pi i y\cdot u}du\\
&=\frac{e^{2\pi iy\cdot V^{-1}x}}{\det(V)}\int_{\bbR^d}e^{-\omega|u|^2}e^{-2\pi i y\cdot V^{-1}u}du,
\end{align*}
where we used \cite[Theorem 2.44]{Folland}.  If we denote the adjoint of a matrix $A$ by $A^*$, then we can rewrite this as
\[
\frac{e^{2\pi i(V^*)^{-1}y\cdot x}}{\det(V)}\int_{\bbR^d}e^{-\omega|u|^2}e^{-2\pi i(V^*)^{-1}y\cdot u}du.
\]
This integral is now just the Fourier transform of a standard Gaussian in $\bbR^d$.  The result is
\begin{align}\label{fhat}
\hat{f}(y)=\frac{e^{2\pi i(V^*)^{-1}y\cdot x}\pi^{d/2}}{\det(V)\omega^{d/2}}e^{-\pi^2|(V^*)^{-1}y|^2/\omega}
\end{align}
(see \cite[Proposition 8.24]{Folland}).  We can now state our desired conclusion.

\begin{PT}
For any $\omega\in(0,\infty)$ and $x\in\bbR^d$, it holds that
\[
\sum_{v\in\mcv}e^{-\omega|x+v|^2}=\frac{\pi^{d/2}}{\det(V)\omega^{d/2}}\sum_{s\in\mcv^*}e^{2\pi is\cdot x}e^{-\pi^2|s|^2/\omega}.
\]
\end{PT}

\end{document}